\documentclass[11pt,reqno]{amsart}
\usepackage{hyperref}
\usepackage{amsmath,amssymb,amsfonts,amsthm,enumerate,natbib,color,ifthen}
\usepackage{xargs}
\usepackage[textwidth=4cm, textsize=footnotesize]{todonotes}
\usepackage[latin1]{inputenc}

\usepackage{adjustbox}
\usepackage{multicol,booktabs,colortbl,tabularx}
\usepackage{caption,subcaption}
\usepackage{float}
\usepackage{graphicx}

\usepackage{aliascnt,bbm}
\def\rset{\mathbb R}
\def\zset{\mathbb Z}

\def\eqsp{\;}

\newcommand{\pscal}[2]{\left\langle#1,#2\right\rangle}

\newcommand{\eqdef}{\ensuremath{\stackrel{\mathrm{def}}{=}}}

\def\Xset{\mathcal{Z}} 
\def\Zset{\mathcal{Z}} 
\def\F{\mathcal{F}} 
\def\cB{\mathsf{B}} 
\def\e{\mathcal{E}}

\def\q{\mathsf{q}}

\def\A{\mathcal{A}}

\newcommandx\sequence[3][2=t,3=\zset]
{\ifthenelse{\equal{#3}{}}{\ensuremath{\{ #1_{#2}\}}}{\ensuremath{\{ #1_{#2}, \eqsp #2 \in #3 \}}}}

\def\PP{\mathbb{P}} 
\newcommand{\CPP}[3][]
{\ifthenelse{\equal{#1}{}}{{\mathbb P}\left(\left. #2 \, \right| #3 \right)}{{\mathbb P}_{#1}\left(\left. #2 \, \right | #3 \right)}}
\def\PE{\mathbb{E}} 
\newcommand{\CPE}[3][]
{\ifthenelse{\equal{#1}{}}{{\mathbb E}\left[\left. #2 \, \right| #3 \right]}{{\mathbb E}_{#1}\left[\left. #2 \, \right | #3 \right]}}

\def\Cset{\mathcal{C}} 

\usepackage{bm}

\theoremstyle{plain}
\newtheorem{theorem}{Theorem}
\newtheorem{assumption}{H\hspace{-3pt}}

\newaliascnt{proposition}{theorem}

\aliascntresetthe{proposition}

\newaliascnt{lemma}{theorem}
\newtheorem{lemma}[lemma]{Lemma}
\aliascntresetthe{lemma}
\newaliascnt{corollary}{theorem}

\aliascntresetthe{corollary}

\theoremstyle{definition}
\newaliascnt{definition}{theorem}

\aliascntresetthe{definition}

\newaliascnt{remark}{theorem}
\newtheorem{remark}[remark]{Remark}
\aliascntresetthe{remark}

\newaliascnt{example}{theorem}

\aliascntresetthe{example}

\def\rmd{\mathrm{d}}

\def\1{\mathbbm{1}}

\usepackage{hyperref}
\usepackage{subcaption} 
\begin{document}

\title[Linear regression models with instrumental variables]{Bayesian variable selection in linear regression models with instrumental variables}\thanks{This work is partially supported by the NSF grant DMS 1513040}

\author{Gautam Sabnis}\thanks{ G. Sabnis: Boston University, 111 Cummington Mall, Boston,
  02215, MA, United States. {\em E-mail address:} gsabnis@bu.edu}
\author{Yves Atchad\'e}\thanks{ A. Y. Atchad\'e: Boston University, 111 Cummington Mall, Boston,
  02215, MA, United States. {\em E-mail address:} atchade@bu.edu}
\author{Prosper Dovonon}\thanks{O. Dovonon: Concordia University, 1455 de Maisonneuve Blvd. West
Montreal, Quebec, H3G 1M8, Canada. {\em E-mail address:} prosper.dovonon@concodrdia.edu}

\subjclass[2010]{62F15, 62Jxx}

\keywords{High-dimensional Bayesian inference, Endogeneity, Variable selection, Posterior contraction, Markov Chain Monte Carlo, linear regression}

\maketitle


\begin{abstract}
Many papers on high-dimensional statistics have proposed methods for variable selection and inference in linear regression models by relying explicitly or implicitly on the assumption that all regressors are exogenous. However, applications abound where endogeneity arises from selection biases, omitted variables, measurement errors, unmeasured confounding and many other challenges common to data collection \citep{fan2014challenges}. The most common cure to endogeneity issues consists in resorting to instrumental variable (IV) inference. The objective of this paper is to present a Bayesian approach to tackling endogeneity in high-dimensional linear IV models. Using a working quasi-likelihood combined with an appropriate sparsity inducing spike-and-slab prior distribution, we develop a semi-parametric method for variable selection in high-dimensional linear models with endogeneous regressors within a quasi-Bayesian framework. We derive some conditions under which the quasi-posterior distribution is well defined and puts most of its probability mass around the true value of the parameter as $p \rightarrow \infty$. We demonstrate through empirical work the fine performance of the proposed approach relative to some other alternatives. We also include include an empirical application that assesses the return on education by revisiting the work of \citet{angrist1991does}.  
\end{abstract}

\section{Introduction}
The linear regression model has imposed itself as a benchmark for assessing the relationship between a response variable of interest and a set of covariates, or regressors. A critical issue in regression models is that of endogeneity, that is when a subset of regressors is correlated with the regression model error.
 Basically endogenous variables are those influenced by some of the same forces that influence the response variable. For example, economists examining the effects of education on earnings have long been concerned with the endogeneity of education \citep{angrist1991does}. ``Ability'' is often cited as one factor possibly correlated with earnings (those with higher ability earn more) and education (those with higher ability obtain more education). Endogeneity also arises from measurement errors in the explanatory variables. It is well-known in regressions with small set of regressors that endogeneity causes standard estimators such as the ordinary least squares estimator to be inconsistent.


The most common cure to endogeneity issues consists in resorting to instrumental variable (IV) inference. Consistent estimation is commonly obtained by relying on the so-called valid  instrumental variables (IV); i.e. variables uncorrelated with the regression error but correlated with the endogenous regressors. This gives rise to the IV model:
 \begin{equation*}
\PE[w_{ik}(y_i-x_i^\prime \theta)] = 0,\quad k = 1, \ldots, q, \quad i = 1, \ldots, n
\end{equation*}
where $y_i\in \rset$ is the response variable, $x_i\in \rset^p$ is the vector of explanatory variables, $w_i \in \rset^q$ the vector of instruments, $\theta\in \rset^p$ the vector of parameters, and $n$ is the sample size. A good account of the IV methodology in low-dimensional problems can be found in \citet{angrist1991does,hansen1982large}, and the references therein. 

In this paper we consider high-dimensional linear regression models where the number of regressors $p$ is potentially larger than the sample size $n$. This set up is not immune to endogeneity. In fact, beside the usual sources mentioned above,  in some settings endogeneity can  arise incidentally from a large number of regressors (see e.g. \citet{fan2014endogeneity}). 
 Recent work related to high-dimensional inference on linear IV models include Belloni et al. (2012), \citet{gautier2014high, fan2014endogeneity, belloni2017simultaneous}.
 Belloni et al. (2012) propose a two-step lasso/post-lasso approach for instrument selection and inference in linear IV models where the number of explanatory variables ($p$) is fixed but the number of instrumental variables ($q$) is large. \citet{gautier2014high} consider $p$ large and possibly $q$ large and propose the so-called Self-tuning IV estimator and non-asymptotic confidence intervals based on the Dantzig selection of \citet{candes2007dantzig}. \citet{belloni2017simultaneous} consider $p$ and $q$ large and propose estimators and confidence regions that are honest and asymptotically correct by relying on a two-step procedure using suitably orthogonalized instruments. \cite{fan2014endogeneity} follows the generalized method of moments (GMM) approach, as introduced by \citet{hansen1982large}. However, when $q\geq n$, the GMM objective function is too noisy an estimator of its population version. This has led \cite{fan2014endogeneity} to propose the focused GMM (FGMM), which  minimizes a GMM criterion that ignores the non-selected regressors. 

This paper relies on GMM settings and proposes a Bayesian method for variable selection and inference in high-dimensional IV models. One of the key advantages of the Bayesian framework is the ability to easily perform inference on the parameters of the model, and incorporate existing prior information in the analysis. By only restricting the moments of the data, IV models obviate the need to assume an underlying data distribution (or complete specification of a likelihood function), and allow inferences about the parameter of interest based only on the partial information supplied by a set of moment conditions. One interesting development in the Bayesian literature over the past few years is the quasi-Bayesian framework, which allows the development of Bayesian procedures without a complete specification of a likelihood function \citep{chernozhukov2003mcmc, liao2011posterior,kato2013quasi,atchade2017contraction} and makes it possible to effectively develop semi-parametric models, and moment equation models.

 The main contributions of this paper are threefold. First, using a working quasi-likelihood combined with an appropriate sparsity inducing spike-and-slab prior distribution \citep{mitchell1988bayesian, george1997approaches}, we develop a semi-parametric method for variable selection in high-dimensional linear models with endogeneous regressors within a quasi-Bayesian framework. Second, we study the statistical properties of the quasi-posterior distribution, $\check\Pi_\gamma$ (defined later in \ref{post:dist}), as the dimension $p$ increases. Under some minimal assumptions, we show (see Theorem \ref{thm1}) that $\check\Pi_\gamma$  puts most of its probability mass around the true value of the parameter as $p \rightarrow \infty$. Third, we develop a practical and efficient Markov Chain Monte Carlo algorithm to sample from $\check\Pi_\gamma$.  To the best of our knowledge, ours is the first paper to study in detail the Bayesian approach to tackling endogeneity in high-dimensional linear IV models. The performance of the Bayesian IV methods is highlighted by Monte Carlo simulations. The paper also includes an empirical application that assesses the return on education using US data by revisiting the work of \citet{angrist1991does}.

 The rest of the paper is organized as follows. The model and the Bayesian method proposed are presented in Section 2. This section also presents our main results establishing the consistency of the selection method proposed. The MCMC sampling algorithm is introduced in Section 3 which also contains our simulation results. Section 4 contains the empirical application  and concluding remarks are included in Section 5.

\subsection{Notation} For an integer $d \ge 1$, we equip the Euclidean space $\rset^d$ with its usual Euclidean inner product $\pscal {\cdot}{\cdot}$, associated norm $\| \cdot \|_2$, and its Borel sigma-algebra. We set $\Delta_d\eqdef\{0,1\}^d$. We will also use the following norms on $\rset^d$: $\|\theta\|_1 \eqdef \sum\limits_{j=1}^{d} \lvert \theta_j \rvert$, $\|\theta\|_0 \eqdef \sum\limits_{j=1}^{d} {\bm{1}}_\{\lvert \theta_j \rvert > 0\}$ and $\|\theta\|_\infty \eqdef \underset{1 \le j \le d}{\mbox{max}} \lvert \theta_j \rvert$. 

For $\delta \in \Delta_d$, we set $\delta^c \eqdef 1 - \delta$, that is $\delta_j^c \eqdef 1 - \delta_j, 1 \le j \le d$. For $\theta in \rset^d$, the sparsity structure of $\theta$ is the element $\delta \in \Delta_d$ defined as $\delta_j =  {\bm{1}}_\{\lvert{\theta_j}\rvert > 0\}, 1 \le j \le d$. 

Throughout the paper $e$ denotes the Euler number and $[d]$ represents $1, \ldots, d$.

\section{Model and Method}
Suppose that we have $n$ independent subjects, and observe on subject $i$ the random vector $(y_i,x_i,w_i)\in\rset\times \rset^p\times \rset^q$. We postulate the following model: for some $\theta\in\rset^p$,
\begin{equation}\label{reg:mod}
y_i = \pscal{x_i}{\theta} + \epsilon_i,\end{equation}
for some zero-mean (un-observable) real-valued random variable $\epsilon_i$. The regression parameter $\theta\in\rset^p$ is the quantity of interest. We consider the setting where $p\geq n$. This problem has attracted an impressive literature over the last two decades, and it is now well-known that the regression parameter $\theta_\star$ can be recovered if it is sparse -- or close to be sparse -- under appropriate assumptions on the regression matrix (see e.g. \cite{buhlGeer11,hastie:etal:15} and the references therein). 
In this work we consider the setting where some of the components of the regressor $x_i$ are endogeneous, in the sense that there are correlated with the error $\epsilon_i$, so that $\PE(\epsilon_ix_i)\neq 0$. As documented in the introduction, this issue is very common in applications, and it is well-known that standard inferential procedures that ignore endogeneity are inconsistent in general.  A well-established approach to mitigate endogeneity is to use instrumental variables.  This is the approach taken here, and the set of instruments at our disposal is $w_i\in\rset^q$. More precisely we make the following data-ganerating assumption.

\begin{assumption}\label{H1}
$\{(y_i,x_i,w_i,\epsilon_i),\;1\leq i \leq n\}$ are $n$ independent and identically distributed  random vectors, where $(y_i,x_i,w_i,\epsilon_i)\in\rset\times\rset^p\times\rset^q\times\rset$, and there exists $\theta_\star$ such that $y_i=\pscal{x_i}{\theta_\star}+\epsilon_i$, for all $i=1,\ldots,n$.
Furthermore we assume that $\epsilon \eqdef (\epsilon_1,\ldots,\epsilon_n)$ is conditionally sub-Gaussian in the sense that there exists $\sigma_0>0$ such that for all $u\in\rset^n$, 
\begin{equation}\label{sub:gaussian:eq}
\PE(\epsilon\vert w) = 0,\;\;\mbox{ and }\;\; \PE\left(e^{\pscal{u}{\epsilon}}\vert w\right)\leq e^{\frac{\sigma_0^2\|u\|_2^2}{2}},
\end{equation}
almost surely, where $w \eqdef(w_1,\ldots,w_n)$.
\end{assumption}

\medskip
Although not explicitly stated in H\ref{H1}, it is expected that the intruments $w_i$ are correlated to the endogeneous components of $x_i$, and this correlation together with (\ref{sub:gaussian:eq})  are leveraged to derive better behaved inference. This is classically done via the GMM estimator  that minimizes 
\[\left(y-X\theta\right)'WDW'\left(y-X\theta\right),\]
or penalized versions thereof, where $y=(y_1,\ldots,y_n)'\in\rset^n$, $X\in\rset^{n\times p}$ has rows $x_i'$, $W\in\rset^{n\times q}$ has rows $w_i'$, and $D\in\rset^{q\times q}$ symmetric positive definite weight matrix. However in a context where $q$ and $p$ are potentially larger than $n$, this approach of using all the intruments may not work, because the GMM functional could be too noisy estimate of its  population version. To circumvent this problem \cite{fan2014endogeneity} proposed the idea of  focused GMM that incorporates a moment selection step: only instruments associated to selected regression parameters are included in the model. 
Note here that the idea of moment selection differs from previous works on moments selection (as in for instance \cite{}) which deal with the question of how to retain only valid moment conditions. In our case, all the moments conditions are assumed valid, but we face the challenge of having too many of them, given the available sample size.  The purpose of this work is to develop a Bayesian version of focused GMM.

Let $\Delta\eqdef \{0,1\}^p$, $\Zset \eqdef\rset^n\times\rset^{n\times p}\times \rset^{n\times q}$.  For $\delta\in\Delta$, $z=(y,X,W)\in\Zset$, we define
\[q_{\delta,\theta}(z) \eqdef \exp\left[-\frac{1}{2}\left(y-X\theta\right)'W \Lambda_{\delta}W'\left(y-X\theta\right)\right],\]
for some diagonal matrix $\Lambda_\delta\in\rset^{q\times q}$ with nonnegative diagonal elements.  We make the following assumption on the prior distribution of $(\delta,\theta)$.
\begin{assumption}\label{H2}
For $\bar s\geq \|\theta_\star\|_0$, and some absolute constant $u>0$, 
\[\omega_\delta \propto \mathsf{q}^{\|\delta\|_0}(1-\mathsf{q})^{p-\delta}\textbf{1}_{\Delta_{\bar s}}(\delta),\]
where $\mathsf{q} =\frac{1}{p^{u+1}}$, and $\Delta_{\bar s} \eqdef\{\delta\in\Delta:\;\|\delta\|_0\leq \bar s\}$.
Furthermore,  given $\delta$ the components of $\theta$ are independent and
\[ \theta_j\vert \delta \sim \left\{\begin{array}{ll} \textbf{N}\left(0,\frac{1}{\rho^2}\right),\;\; \mbox{ if }\;\; \delta_j=1\\ \textbf{N}\left(0,\gamma\right),\;\; \mbox{ if }\;\; \delta_j=0\end{array}\right.\]
 for constants $\rho>0,\gamma>0$ that we specify later in Theorem \ref{thm1}.
\end{assumption}

\begin{remark}
Discrete priors distributions that put independent Bernoulli distribution on each $\delta_j$ are common in Bayesian variable selection problems (\cite{ george1997approaches}). Note here however that the probability parameter $\textsf{q}$ depends on the dimension $p$. As shown in (\cite{castillo:etal:12}), this feature is key to achieve posterior consistency as $p$ diverges.

Since our objective at the onset is to fit a sparse model, the idea of imposing a hard constrain $\|\delta\|_0\leq \bar s$ on the sparsity level seems reasonable, and has been explored by others (see for instance \cite{banerjee:ghosal13}). The parameter $\bar s$ needs not be a good estimate of $s_\star$, but rather an upper bound derived for instance from prior information or from limitation imposed by the available sample size. 
\end{remark}

Let  $B_\delta\in\rset^{p\times p}$ be the diagonal matrix such that $B_{\delta,jj}=\frac{1}{\rho^2}$ if $\delta_j=1$, and $B_{\delta,jj}=\gamma$ if $\delta_{jj}=0$. Under assumptions H\ref{H1} and H\ref{H2}, the posterior distribution of $(\delta,\theta)$ can be written as 
\begin{equation}\label{post:dist}
\check\Pi(\delta,\rmd\theta\vert z)\propto \omega_\delta q_{\delta,\theta_\delta}(z)\frac{e^{-\frac{1}{2}\theta' B_\delta^{-1}\theta}}{\sqrt{\det(2\pi B_\delta)}}\rmd \theta,\end{equation}
that we view as a random probability measure on $\Delta\times \rset^p$, and we derive in Theorem \ref{thm1} some simple conditions under which $\check\Pi(\cdot\vert z)$ put most of its probability mass around $(\delta_\star,\theta_\star)$, where $\delta_\star$ denotes the sparsity structure of $\theta_\star$, that is $\delta_{\star j} = \textbf{1}(|\theta_{\star j}|>0)$.

Without any loss of generality we will assume that 
\begin{equation}\label{eq:inst:norm}
\|W_j\|_2=1, \;\;\;1\leq j\leq q.\end{equation}
where $W_j$ denotes the $j$-th column of $W$, and we assume that the matrix $\Lambda_\delta$ takes the form
\begin{equation}\Lambda_\delta \eqdef \frac{1}{\lambda} \left(\begin{array}{ccc} (T_\delta)_1 & & \\& \ddots & \\ & & (T_\delta)_q \end{array}\right)\in\rset^{q\times q} 
\end{equation}
for some constant $\lambda>0$, where $T_\delta\eqdef ((T_\delta)_1,\ldots,(T_\delta)_q)\in\{0,1\}^q$, and $(T_\delta)_j=1$ if the $j$-th instrument is included with model $\delta$, $(T_\delta)_j=0$ otherwise. We will write  $A_j$ to denote the $j$-th column of the matrix $A$. And in the same vein, since $T_\delta\in\{0,1\}^q$, we will write $W_{\delta}$ to denote the submatrix of $W$ obtained by keeping only the columns of $W$ for which the corresponding component of  $T_\delta$ is $1$. Under the prior distribution assumption H\ref{H2}, the maximum number of instruments used in any given model is
\begin{equation}\label{def:bart}
\bar t \eqdef \max_{\delta\in\Delta_{\bar s}} \|T_\delta\|_0
\end{equation}
which is expected to be of the same order as $\bar s$, the maximum number of active regressors allowed under prior H\ref{H1}. The matrix 
\[ M_\delta\eqdef (W_{\delta})'X\in\rset^{\|T_\delta\|_0\times p},\]
plays an important role in the analysis.  Its restricted eigenvalues are defined as follows. For  $\delta\in\Delta$ we define
\[\bar v(\delta) \eqdef\sup\left\{ \frac{u'(M_\delta'M_\delta)u}{n\|u\|_2^2},\;u\neq 0, u\in\rset^p_\delta\right\},\]
and
\[\underline{v}(\delta) \eqdef\inf\left\{ \frac{u'(M_\delta'M_\delta)u}{n\|u\|_2^2},\;u\neq 0,u\in\rset^p_\delta \right\}.\]
Note that these quantities depend on the random variable $z$. 

\begin{theorem}\label{thm1}
Assume H\ref{H1}-H\ref{H2}. Choose constants $\bar\kappa_1$, $\bar \kappa\geq 0$, $\underline{\kappa}>0$, and set 
\begin{multline}
\e \eqdef\left\{(y,X,W)\in\Xset:\; \max_{\1\leq k\leq q }\left|\pscal{W_k}{y-X\theta_\star}\right|\leq \sigma_0\sqrt{2\log(pq)},\;\;\bar v(\delta_\star)\leq \bar\kappa\right.\\
\left.\;\;\inf_{\delta\in\Delta_{\bar s}}\underline{v}(\delta)\geq \underline{\kappa} \;\;\;\mbox{ and }\;\;  \max_{\delta\in\Delta_{\bar s}}\max_{1\leq j\leq p} \frac{1}{\sqrt{n}}\left\|W_{\delta}'X_j\right\|_2\leq \bar\kappa_1\right\}.\end{multline}
Choose $\gamma>0$, $\rho\geq 1$ such that $\rho^2\|\theta_\star\|_\infty\leq \bar \rho$, where
\begin{equation}\label{eq:rho}
\bar\rho \eqdef 2\sigma_0\frac{\kappa_1}{\lambda}\sqrt{2n\bar t\log(pq)}.
\end{equation}
Set
\begin{equation}
\epsilon \eqdef 2\sqrt{2}\sigma_0 \frac{\bar\kappa_1}{\underline{\kappa}} \sqrt{\frac{(\bar s + s_\star)\bar t\log(pq)}{n}},
\end{equation}
and for absolute constants $m>1$, $M>\max(u,128)$, set
\[\cB_{m,M}\eqdef \bigcup_{\delta\in\Delta_{\bar s}}\; \{\delta\}\times \left\{\theta\in\rset^p:\; \|\theta_\delta-\theta_\star\|_2 \leq M\epsilon,\;\|\theta-\theta_\delta\|_2\leq m\sqrt{\gamma p}\right\}.\]
Then for all $p$ large enough,  we have
\begin{equation}\label{eq:main:bound}
1 - \PE_\star\left[\check\Pi\left(\cB_{m,M}\vert z\right)\right] \leq  \PP_\star(z\notin \e)  + \frac{1 + (pq)^{\frac{\sigma_0^2\bar t}{\lambda}}}{p^{M^2(1+s_\star)}} + 2e^{-\left(\frac{m-1}{2}\right)p}.
\end{equation}
\end{theorem}
\begin{proof}
See Section \ref{sec:proof:thm1}.
\end{proof}

In general $\check\Pi$ cannot achieve perfect model recovery since the non-zero components of $\theta_\star$ could be arbitrarily small, and hence easily missed. With $C$ and $\epsilon$ as above we define
\[J_\star\eqdef\left\{1\leq j\leq p:\; |\theta_{\star j}|>M\epsilon\right\}.\]
Set $\cB^{(\delta)} \eqdef \left\{\theta\in\rset^p:\; \|\theta_\delta-\theta_\star\|_2 \leq M\epsilon,\;\|\theta-\theta_\delta\|_2\leq m\sqrt{\gamma p}\right\}$. Then, clearly the set $\cB_{m,M}$ of Theorem \ref{thm1} can also be written as
\[\cB_{m,M} = \bigcup_{\delta\in\A} \{\delta\}\times \cB^{(\delta)},\;\;\mbox{ where } \A\eqdef\{\delta\in\Delta:\;\|\delta\|_0\leq \bar s,\;\mbox{ and }\delta_j=1 \mbox{ for all } j\in J_\star\}.\]
In other words, Theorem \ref{thm1} implies that $\check\Pi$ does not miss any of the large magnitude components of $\theta_\star$.

\begin{remark}
\begin{itemize}
\item The result implies that we should let $\lambda$ scale as 
\[\lambda = c \log(pq),\]
for some tuning constant $c$, and then choose $\rho$ as
\[\rho  = c\left(\frac{n}{\log(pq)}\right)^{1/4},\]
for some tuning parameter $c$. Finally the result suggests setting
\[\gamma =\frac{c}{p},\]
for a tuning parameter $c$.
\item The key parameter in the theorem is $\underline{\kappa}$ which depends both on the design matrix $X$ and on the strength of the instruments. For instance, suppose that we are in the situation where one of the instruments, say the first instrument, is weak in the sense that
\[\inf\{\|W_\delta'X_1\|_2,\;\;\delta\in\Delta_{\bar s} \;\;\textsf{s.t. }\; \delta_1=1\} \leq \alpha\bar\kappa_1.\]
In that case, if $e_1$ denotes the first unit vector of $\rset^p$, we have $e_1'M_\delta'M_\delta e_1 = \|W_\delta'X_1\|_2$. Hence
\[\frac{\inf_{\delta\in\Delta_{\bar s}}\underline{v}(\delta)}{\bar\kappa_1} \leq \alpha.\]
\end{itemize}
\end{remark}

\section{Markov Chain Monte Carlo computation and numerical experiments}\label{sec:num}
In this section we develop a practical Markov Chain Monte Carlo algorithm to sample from the posterior distribution $\check\Pi$, and we explore the behavior of $\check\Pi$ on two simulated data examples.

\subsection{A MCMC sampler for $\Pi$ \label{mcmc_sampler}}  
We begin with a description of the MCMC sampler. To sample $(\delta, \theta)$ we use a Metropolis-Hastings-within-Gibbs sampler, where we update $\delta$ given $\theta$, then we update the selected component $[\theta]_\delta$ given $(\delta, [\theta]_{\delta^c})$, and finally update $[\theta]_{\delta^c}$ given $(\delta, [\theta]_\delta)$. We refer the reader to \cite{tierney94, robertetcasella04} for introduction to basic MCMC algorithms. 

To update $\delta$, we follow a specific form of Metropolis-Hastings update analyzed in \citep{yang2016computational}. To develop the details we rewrite the posterior in \eqref{post:dist} as follows,
\begin{equation}
\check\Pi(\delta,\rmd\theta\vert z)\propto \omega_\delta e^{-\frac{1}{2} \sum\limits_{\ell = 1}^{q}\frac{{\tilde\delta}_\ell}{v_\ell} {\langle y - X\theta_\delta, w_\ell \rangle}^2} \frac{e^{-\frac{1}{2}\theta' B_\delta^{-1}\theta}}{\sqrt{\det(2\pi B_\delta)}}\rmd \theta,\end{equation}
where $\tilde\delta = T(\delta)$. We randomly select one of the following two schemes to update $\delta$, each with probability 0.5.

{\bf{\textit{Single flip update:}}} Choose an index $j \in [p]$ uniformly at random, and form the new state $\delta'$ by setting $\delta_j' = 1 - \delta_j$. We denote this by $\delta_{(j)} \rightarrow \delta'_{(j)}$. 

Move to the state $\delta'_{(j)}$ with probability Pr($\delta_{(j)}, \delta'_{(j)}$) where the acceptance ratio is given by 
\begin{equation*} 
\mbox{Pr}(\delta_{(j)},\delta'_{(j)}) = \min \bigg\{1, \frac{\Pi(\delta'_{(j)} \mid z)}{\Pi(\delta_{(j)} \mid z)}\bigg\}
\end{equation*} 

For a single flip on $i_0$ where $\delta'_{i_0} = 1$,
\begin{align*}
\frac{\Pi(\delta'_{(i_0)} \mid z)}{\Pi(\delta_{(i_0)} \mid z)} = 
\frac{q N(0, c/\rho)}{(1-q) N(0, \gamma)} \exp\bigg\{-\frac{1}{2} \sum\limits_{\ell \in {\tilde{i}}_0} \frac{1}{v_\ell}{\langle y - X\theta_{\delta}, w_{\ell}\rangle}^2 - \frac{1}{2} \theta_{i_0}^2 \sum\limits_{\ell=1}^{q} \frac{\tilde{\delta'}_\ell}{v_\ell} {\langle X_{i_0}, w_j \rangle}^2 \\ 
+ \theta_{i_0} \sum\limits_{\ell=1}^{q} \frac{\tilde{\delta'}_\ell}{v_\ell} {\langle y - X\theta_{\delta}, w_\ell \rangle} {\langle X_{i_0}, w_\ell \rangle}\bigg\}
\end{align*} 

where $\tilde\delta = (\delta, \delta)$, $\tilde{i} = (i, p+i)$ and $y - X\theta_\delta = y - X_{-i_0}\theta_{\delta_{-i_0}}$. 

For a single flip on $i_0$ where $\delta'_{i_0} = 0$,
\begin{align*}
\frac{\Pi(\delta'_{(i_0)} \mid z)}{\Pi(\delta_{(i_0)} \mid z)} = 
\frac{(1 - q) N(0, \gamma)}{q N(0, c/\rho)} \exp\bigg\{\frac{1}{2} \sum\limits_{\ell \in {\tilde{i}}_0} \frac{1}{v_\ell} {\langle y - X\theta_{\delta}, w_{\ell}\rangle}^2 + \frac{1}{2} \theta_{i_0}^2 \sum\limits_{\ell=1}^{q} \frac{\tilde{\delta}_\ell}{v_\ell} {\langle X_{i_0}, w_\ell \rangle}^2\\ - \theta_{i_0} \sum\limits_{\ell=1}^{q} \frac{\tilde{\delta}_\ell}{v_\ell} {\langle y - X\theta_{\delta}, w_\ell \rangle} {\langle X_{i_0}, w_\ell \rangle}\bigg\}
\end{align*}

where $y - X\theta_\delta = y - X_{-i_0}\theta_{\delta_{-i_0}}$.

{\bf{\textit{Double flip update:}}} Define the subsets $S(\delta) = \{j \in [p] \mid \delta_j = 1\}$ and let $S^{c}(\delta) = \{j \in [p] \mid \delta_j = 0\}$. Choose an index pair $(j_1, j_2) \in S(\delta) \times S^c(\delta)$ uniformly at random, and form the new state $\delta'$ by flipping $\delta_{j_1} = 1$ to $\delta'_{j_1} = 0$ and $\delta_{j_2} = 0$ to $\delta'_{j_2} = 1$. We denote this by $\delta_{(j_1,j_2)} \rightarrow \delta'_{(j_1,j_2)}$.  

Move to the state $\delta'_{(j_1,j_2)}$ with probability Pr($\delta_{(j_1,j_2)}$,$\delta'_{(j_1,j_2)}$) where the acceptance ratio is given by 
\begin{equation*} 
\mbox{Pr}(\delta_{(j_1,j_2)},\delta'_{(j_1,j_2)}) = \min \bigg\{1, \frac{\Pi(\delta'_{(j_1,j_2)} \mid z)}{\Pi(\delta_{(j_1,j_2)} \mid z)}\bigg\}
\end{equation*} 
For a double flip on $i_0$ and $i_1$ where $\delta'_{i_0} = 0$ and $\delta'_{i_1} = 1$, 

\begin{align*}
\frac{\Pi(\delta'_{(i_0,i_1)} \mid z)}{\Pi(\delta_{(i_0,i_1)} \mid z)} = 
\exp\bigg\{-\frac{1}{2} \theta_{i_1}^2 \sum\limits_{\ell \ne {\tilde{i}}_0} \frac{\tilde{\delta'}_\ell}{v_\ell} {\langle X_{i_1}, w_\ell \rangle}^2 + \frac{1}{2} \theta_{i_0}^2 \sum\limits_{\ell \ne {\tilde{i}}_1} \frac{\tilde{\delta}_\ell}{v_\ell} {\langle X_{i_0}, w_\ell \rangle}^2 \\ 
+ \theta_{i_1} \sum\limits_{\ell \ne {\tilde{i}}_0} \frac{\tilde{\delta'}_\ell}{v_\ell} {\langle y - X\theta_\delta, w_\ell \rangle}{\langle X_{i_1}, w_\ell \rangle} \\ - \theta_{i_0} \sum\limits_{\ell \ne {\tilde{i}}_1} \frac{\tilde{\delta}_\ell}{v_\ell} {\langle y - X\theta_\delta, w_\ell \rangle}{\langle X_{i_0}, w_\ell \rangle} - \frac{1}{2} \sum\limits_{\ell \in {\tilde{i}}_1} \frac{1}{v_\ell} {\langle y - X\theta_\delta, w_{\ell}\rangle}^2 + \frac{1}{2} \sum\limits_{\ell \in {\tilde{i}}_0} \frac{1}{v_\ell}{\langle y - X\theta_\delta, w_{\ell} \rangle}^2 \bigg\}
\end{align*} 

$y - X\theta_\delta = y - X_{-\{i_0,i_1\}}\theta_{\delta_{-\{i_0,i_1\}}}$.

The full conditionals of $\theta$ are standard distributions due to the use of Gaussian prior. We partition $\theta$ into $\theta = ([\theta]_{\delta}, [\theta]_{\delta^c})$, where $[\theta]_\delta$ groups the components of $\theta$ for which $\delta_j = 1$, and $[\theta]_{\delta^c}$ groups the remaining components. The conditional distributions of the two components are given by 
\begin{equation*} 
\check\Pi(\theta_\delta \mid \delta, z) \sim N\Bigg(\big(X'_\delta W \Lambda_\delta W' X_\delta + B_\delta^{-1} \big)^{-1} X'_\delta W \Lambda_\delta W' Y, \big(X'_\delta W \Lambda_\delta W' X_\delta + B_\delta^{-1}\big)^{-1}\Bigg) 
\end{equation*} 

\begin{equation*} 
\check\Pi(\theta_{\delta^c} \mid \delta^{c}, z) \sim N(0, B_{\delta^c}^{-1}) 
\end{equation*}

\subsection{Numerical Experiments}
In this section we investigate the performance of our proposed approach via numerical simulations, using the same set up as in \citet{fan2014endogeneity,belloni2017simultaneous}.
We simulate from a linear model 
\begin{equation*} 
Y = X^{T} \theta_0 + \epsilon
\end{equation*} 

For each component of $X$, we write $X_j = X_j^e$ if $X_j$ is endogeneous, and $X_j = X_j^x$ if $X_j$ is exogeneous. $X_j^e$, $X_j^x$ and $\epsilon$ are generated according to two different setups which we outline below. 

{\bf{Setup 1:}}

\begin{equation*} 
X_j^e = (F_j + H_j + 1)(3\epsilon + 1), \quad X_j^x = F_j + H_j + u_j
\end{equation*} 
where $\{\epsilon, u_1, \ldots, u_p\}$ are independent $N(0,1)$. Here $F = (F_1, \ldots, F_p)^{T}$ and $H = (H_1, \ldots, H_p)^{T}$ are the transformations of a three-dimensional instrumental variable $V = (V_1, V_2, V_3)^{T} \sim N(0, \mathrm{I}_3)$ and $W = (F,H)$. There are $m$ endogeneous variables $(X_1, X_2, X_3, X_6, \ldots, X_{2+m})^{T}$ with $m = \{10,50\}$.

The Fourier basis are applied as the working instruments,  
\begin{align*} 
F = \sqrt{2}\{\sin(j\pi V_1) + \sin(j \pi V_2) + \sin(j \pi V_3): j \le p\} \\
H = \sqrt{2}\{\cos(j\pi V_1) + \cos(j \pi V_2) + \cos(j \pi V_3): j \le p\}
\end{align*} 

{\bf{Setup 2:}} 
\begin{equation*}
X_j^e = \tilde{X_j} + \sum\limits_{t = 1}^{T} z_{T(j - 1) + t}, \quad \epsilon = \zeta + {\tilde{X}}^{'} \gamma_0
\end{equation*}

where $\gamma0 = (.1,.2,.3,\ldots,1,0,\ldots)^{'}$, $z \sim N(0, I_{Tp})$, ${\tilde{X}} \sim N(0, \Sigma)$, $\Sigma_{ij} = 0.3^{\mid{i-j}\mid}$, and $\zeta \sim N(0, 1/4^2)$.

The two setups are taken from \citet{fan2014endogeneity} and \citet{belloni2017simultaneous} respectively. For both setups, we choose the design vector $\theta_0 \in {\mathbb{R}}^{p}$ with number of non-zero components, $s_\star = 5$, that takes the value 

\begin{equation*}
\theta_\star = \mbox{SNR} \times (5, -4, 7, -2 ,1.5, 0, \ldots, 0)'
\end{equation*} 
where $\mbox{SNR} > 0$ is a signal-to-noise parameter. Varying the $\mbox{SNR}$ parameter allows us to explore the performance of our approach for varying levels of signal strength.  We performed simulations for $\mbox{SNR} = \{0.25, 1\}$, sample size $n = 100$, and number of covariates $p \in \{100,200\}$. $\mbox{SNR} = 1$ corresponds to high $\mbox{SNR}$ ($\mbox{hSNR}$) while $\mbox{SNR} = 0.25$ corresponds to weak $\mbox{SNR}$ ($\mbox{wSNR}$). 

In our experiments, we used 100 replications to aggregate the results. Four performance measures are used to compare the methods. The first measure is the number of correctly identified nonzero coefficients, that is, the true positive (TP). The second measure is the number of incorrectly identified coefficients, the false positive (FP). The last two measures are mean squared errors, ${\mbox{MSE}}_{S}$ \&  ${\mbox{MSE}}_{N}$, of the important and unimportant regressors respectively determined by averaging $\|\hat\theta - \theta_\star\|^2$ on $S = \{1,2,3,4,5\}$ and $N = S^c$ over 100 replications. The standard errors over 100 replications for each measure are also reported. In each run of the MCMC sampler, $\hat\theta$ is initialized using penalized least squares $[\mbox{SCAD}(\lambda_{\mbox{scad}})]$ with $\lambda_{\mbox{scad}} = 1$ and $\delta$ is initialized by setting $\hat\delta^{(0)} = \textbf{1}(|\hat\theta^{(0)}_{j}|>0)$. $\mbox{FGMM}$ results are obtained using the code on the authors' website by setting the $\mbox{FGMM}$ parameter $\lambda_{\mbox{fgmm}} = 0.3$. Our proposed method has three tuning parameters. In all our empirical work, we set 
\begin{equation*} 
\frac{1}{\rho^2} = \frac{\log{(p*q)}}{\sqrt{n}}, \gamma = \frac{10}{p}, \lambda = \begin{cases} n & \mbox{Setup 1} \\
n^{1/3} & \mbox{Setup 2} \end{cases} 
\end{equation*} 
where $q$ is the number of instrumental variables.

The summary of our results is presented in Tables \ref{table:setup1} - \ref{table:setup2} and figure \ref{fig:setup1}. We compare our method, quasi-Bayesian moment restrictions model ($\mbox{BMRM}$), with $\mbox{FGMM}$ and penalized least squares ($\mbox{PLS}$).

In Setup 1 and for the high signal-to-noise regime ($\mbox{SNR} = 1$), $\mbox{PLS}$ performs well in selecting the true coefficients but, at the same time, includes a significantly large number of false positives. $\mbox{FGMM}$ reduces the number of unimportant coefficients while keeping the important coefficients in the model. In contrast, \mbox{BMRM} not only selects all the important coefficients but also succeeds in weeding out almost all the unimportant coefficients. Our proposed method stands out in this regard. Further, the average $\mbox{MSE}_S$ of both $\mbox{FGMM}$ and $\mbox{BMRM}$ is less than that of PLS since the instrumental variables estimation is used for estimating the coefficients. The lower panel of \ref{table:setup1} displays results for the weak signal-to-noise regime ($\mbox{SNR} = 0.25$) case. Again, $\mbox{BMRM}$ outperforms $\mbox{FGMM}$ in selecting the important regressors and removing the unimportant regressors. 

To study the effect of variable selection when the number of endogenous variables is increased, we run another set of simulations with the same data generating process as in table \ref{table:setup1} but we increase $m$ from 10 to 50. Figure \ref{fig:setup1} display our results. It  is clearly seen that $\mbox{BMRM}$ outperforms $\mbox{FGMM}$ and $\mbox{PLS}$.  

In Setup 2 and for $\mbox{hSNR}$ regime, $\mbox{PLS}$ identifies the important covariates but it does so at the cost of overfitting resulting in false discoveries. In terms of TPs, although $\mbox{BMRM}$ does not always outperform its competitors, it remains competitive. When signals are low (lower panel of Table \ref{table:setup2}), all methods under consideration have trouble finding the right model, highlighting the difficulty of identifying the right model with limited sample size. On the other hand, there is some promising news. In all cases, the proposed $\mbox{BMRM}$ method leads to slightly lower false positive rates compared to $\mbox{FGMM}$ and $\mbox{PLS}$.

\begin{table}[H]
\begin{center} 
\caption{Setup 1: Endogeneity in both important and unimportant regressors, n = 100,m = 10, $s_0 = 5$. Top and bottom panels correspond to $\mbox{hSNR}$ and $\mbox{wSNR}$ regimes respectively.}
\begin{flushleft} 
\scalebox{0.9}{
\begin{tabular}{|c|cccc|cccc|cccc|} \toprule
& \multicolumn{4}{c}{\bf{BMRM}} & \multicolumn{4}{c}{\bf{FGMM}} & \multicolumn{4}{c}{\bf{PLS}}\\ 
\cmidrule(lr){2-5} \cmidrule(lr){6-9} \cmidrule(lr){10-13} \\
p & TP & FP & ${\mbox{MSE}}_S$ & ${\mbox{MSE}}_N$ & TP & FP & ${\mbox{MSE}}_S$ & ${\mbox{MSE}}_N$ & TP & FP & ${\mbox{MSE}}_S$ & ${\mbox{MSE}}_N$ \\
\hline
100 & ${\underset{(0)}{5}}$ & ${\underset{(0\cdot11)}{0\cdot11}}$ & ${\underset{(0\cdot002)}{0\cdot004}}$ & ${\underset{(0\cdot001)}{0\cdot002}}$ & ${\underset{(0)}{5\cdot00}}$ & ${\underset{(1\cdot14)}{3\cdot14}}$ & ${\underset{(0\cdot002)}{0\cdot002}}$ & ${\underset{(0)}{0}}$ & ${\underset{(0)}{5}}$ & ${\underset{(14\cdot98)}{59\cdot08}}$ & ${\underset{(0\cdot03)}{0\cdot02}}$ & ${\underset{(0\cdot004)}{0\cdot003}}$ \\

200 &  ${\underset{(0\cdot003)}{4\cdot99}}$ & ${\underset{(0\cdot30)}{0\cdot57}}$ & ${\underset{(0\cdot003)}{0\cdot005}}$ & ${\underset{(0\cdot000)}{0\cdot003}}$ & ${\underset{(0\cdot10)}{4\cdot99}}$ & ${\underset{(1\cdot42)}{3\cdot29}}$ & ${\underset{(0\cdot05)}{0\cdot007}}$ & ${\underset{(0)}{0}}$ & ${\underset{(0)}{5}}$ & ${\underset{(28\cdot62)}{98\cdot48}}$ & ${\underset{(0\cdot22)}{0\cdot15}}$ & ${\underset{(0\cdot02)}{0\cdot01}}$ \\

\hline \hline 

100 & ${\underset{(0\cdot81)}{4\cdot79}}$ & ${\underset{(2\cdot79)}{0\cdot66}}$ & ${\underset{(0\cdot18)}{0\cdot04}}$ & ${\underset{(0\cdot02)}{0\cdot02}}$ & ${\underset{(0\cdot67)}{4\cdot36}}$ & ${\underset{(1\cdot20)}{3\cdot18}}$ & ${\underset{(0\cdot04)}{0\cdot03}}$ & ${\underset{(0\cdot000)}{0\cdot000}}$ & ${\underset{(0\cdot1)}{4\cdot99}}$ & ${\underset{(10\cdot77)}{21\cdot41}}$ & ${\underset{(0\cdot001)}{0\cdot01}}$ & ${\underset{(0\cdot000)}{0\cdot000}}$ \\

200 &  ${\underset{(0\cdot50)}{4\cdot91}}$ & ${\underset{(0\cdot45)}{0\cdot26}}$ & ${\underset{(0\cdot11)}{0\cdot017}}$ & ${\underset{(0\cdot002)}{0\cdot003}}$ & ${\underset{(0\cdot66)}{4\cdot36}}$ & ${\underset{(1\cdot13)}{3\cdot29}}$ & ${\underset{(0\cdot04)}{0\cdot03}}$ & ${\underset{(0)}{0}}$& ${\underset{(0\cdot20)}{4\cdot96}}$ & ${\underset{(16\cdot91)}{30\cdot02}}$ & ${\underset{(0\cdot01)}{0\cdot01}}$ & ${\underset{(0\cdot000)}{0\cdot000}}$ \\

\hline
\end{tabular}}
\end{flushleft} 
\label{table:setup1} 
\end{center} 
\end{table}

\begin{table}[h]
\begin{center} 
\caption{Setup 2: Endogeneity in all regressors, n = 100, T = 2, $s_0 = 5$. Top and bottom panels correspond to $\mbox{hSNR}$ and $\mbox{wSNR}$ regimes respectively.}
\begin{flushleft} 
\scalebox{0.9}{
\begin{tabular}{|c|cccc|cccc|cccc|} \toprule
& \multicolumn{4}{c}{\bf{BMRM}} & \multicolumn{4}{c}{\bf{FGMM}} & \multicolumn{4}{c}{\bf{PLS}}\\ 
\cmidrule(lr){2-5} \cmidrule(lr){6-9} \cmidrule(lr){10-13} \\
p & TP & FP & ${\mbox{MSE}}_S$ & ${\mbox{MSE}}_N$ & TP & FP & ${\mbox{MSE}}_S$ & ${\mbox{MSE}}_N$ & TP & FP & ${\mbox{MSE}}_S$ & ${\mbox{MSE}}_N$ \\
\hline
100 & ${\underset{(0\cdot57)}{4\cdot76}}$ & ${\underset{(1\cdot39)}{1\cdot74}}$ & ${\underset{(0\cdot45)}{0\cdot31}}$ & ${\underset{(0\cdot04)}{0\cdot04}}$ & ${\underset{(0\cdot50)}{4\cdot79}}$ & ${\underset{(1\cdot98)}{2\cdot93}}$ & ${\underset{(0\cdot45)}{0\cdot34}}$ & ${\underset{(0\cdot005)}{0\cdot002}}$ & ${\underset{(0)}{5}}$ & ${\underset{(3\cdot58)}{8\cdot28}}$ & ${\underset{(0\cdot05)}{0\cdot07}}$ & ${\underset{(0\cdot002)}{0\cdot008}}$ \\

200 &  ${\underset{(0\cdot56)}{4\cdot78}}$ & ${\underset{(1\cdot52)}{2\cdot20}}$ & ${\underset{(0\cdot44)}{0\cdot29}}$ & ${\underset{(0\cdot009)}{0\cdot013}}$ & ${\underset{(0\cdot58)}{4\cdot69}}$ & ${\underset{(2\cdot14)}{3\cdot06}}$ & ${\underset{(0\cdot45)}{0\cdot39}}$ & ${\underset{(0\cdot003)}{0\cdot001}}$ & ${\underset{(0)}{5}}$ & ${\underset{(5\cdot68)}{10\cdot70}}$ & ${\underset{(0\cdot08)}{0\cdot10}}$ & ${\underset{(0\cdot002)}{0\cdot005}}$  \\

\hline \hline 

100 & ${\underset{(1\cdot28)}{2\cdot57}}$ & ${\underset{(1\cdot01)}{1\cdot16}}$ & ${\underset{(0\cdot56)}{0\cdot53}}$ & ${\underset{(0\cdot06)}{0\cdot05}}$ & ${\underset{(1\cdot09)}{2\cdot98}}$ & ${\underset{(1\cdot91)}{2\cdot54}}$ & ${\underset{(0\cdot88)}{0\cdot54}}$ & ${\underset{(0\cdot008)}{0\cdot004}}$ & ${\underset{(0\cdot48)}{4\cdot35}}$ & ${\underset{(1\cdot31)}{4\cdot73}}$ & ${\underset{(0\cdot03)}{0\cdot08}}$ & ${\underset{(0\cdot002)}{0\cdot007}}$ \\

200 &  ${\underset{(0\cdot88)}{3\cdot38}}$ & ${\underset{(1\cdot23)}{1\cdot74}}$ & ${\underset{(0\cdot30)}{0\cdot28}}$ & ${\underset{(0\cdot02)}{0\cdot02}}$ & ${\underset{(1\cdot02)}{3\cdot05}}$ & ${\underset{(2\cdot13)}{2\cdot78}}$ & ${\underset{(0\cdot45)}{0\cdot42}}$ & ${\underset{(0\cdot004)}{0\cdot002}}$& ${\underset{(0\cdot52)}{4\cdot25}}$ & ${\underset{(2\cdot37)}{5\cdot61}}$ & ${\underset{(0\cdot04)}{0\cdot09}}$ & ${\underset{(0\cdot001)}{0\cdot004}}$ \\

\hline
\end{tabular}}
\end{flushleft} 
\label{table:setup2} 
\end{center} 

\end{table}

\begin{figure}[t]
    \centering 
\caption*{$(n,p) = (100,100)$, $m = 50$ and $\theta_0 = [5,-4,7,-2,1.5,{\bm{0}}]'$}
\begin{subfigure}{0.25\textwidth}
  \includegraphics[width=\linewidth]{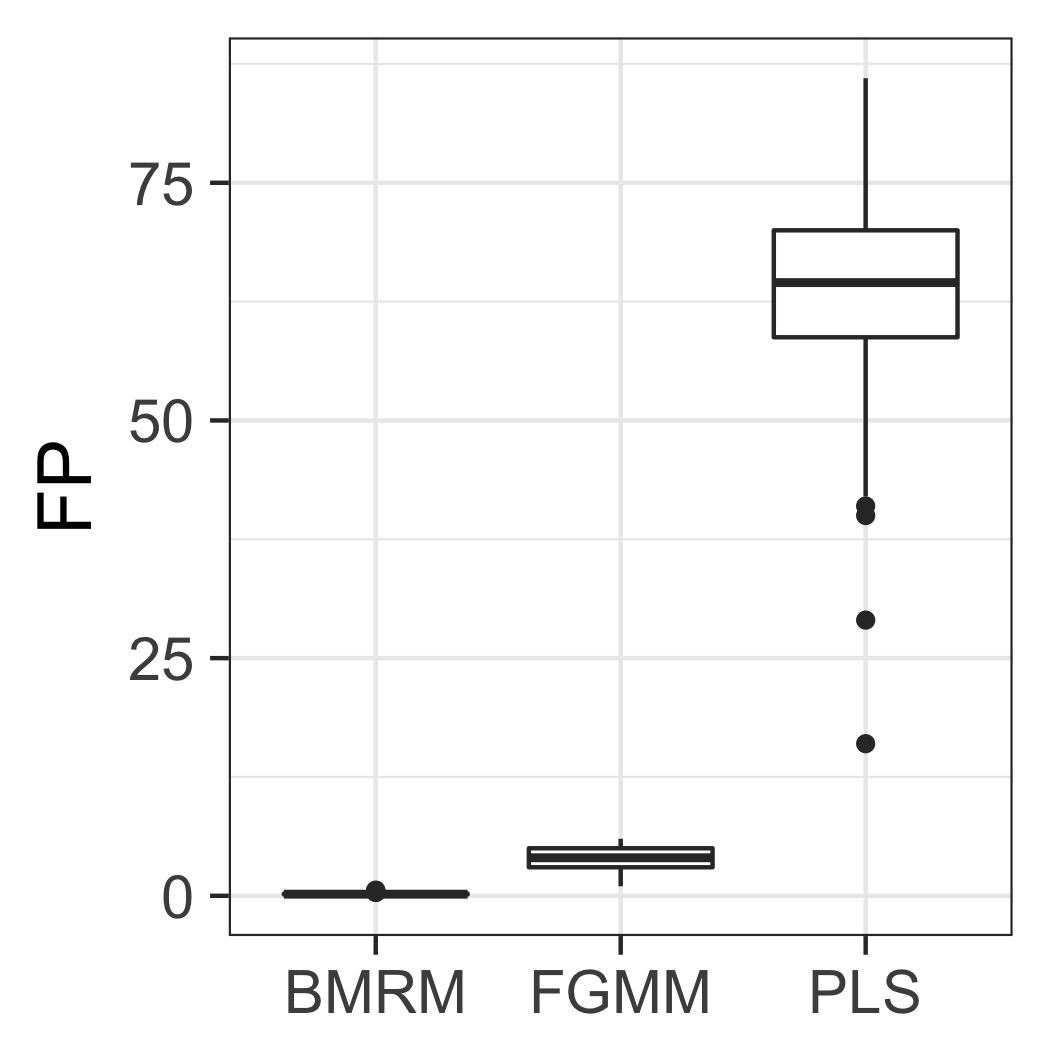}
\end{subfigure}\hfil 
\begin{subfigure}{0.25\textwidth}
  \includegraphics[width=\linewidth]{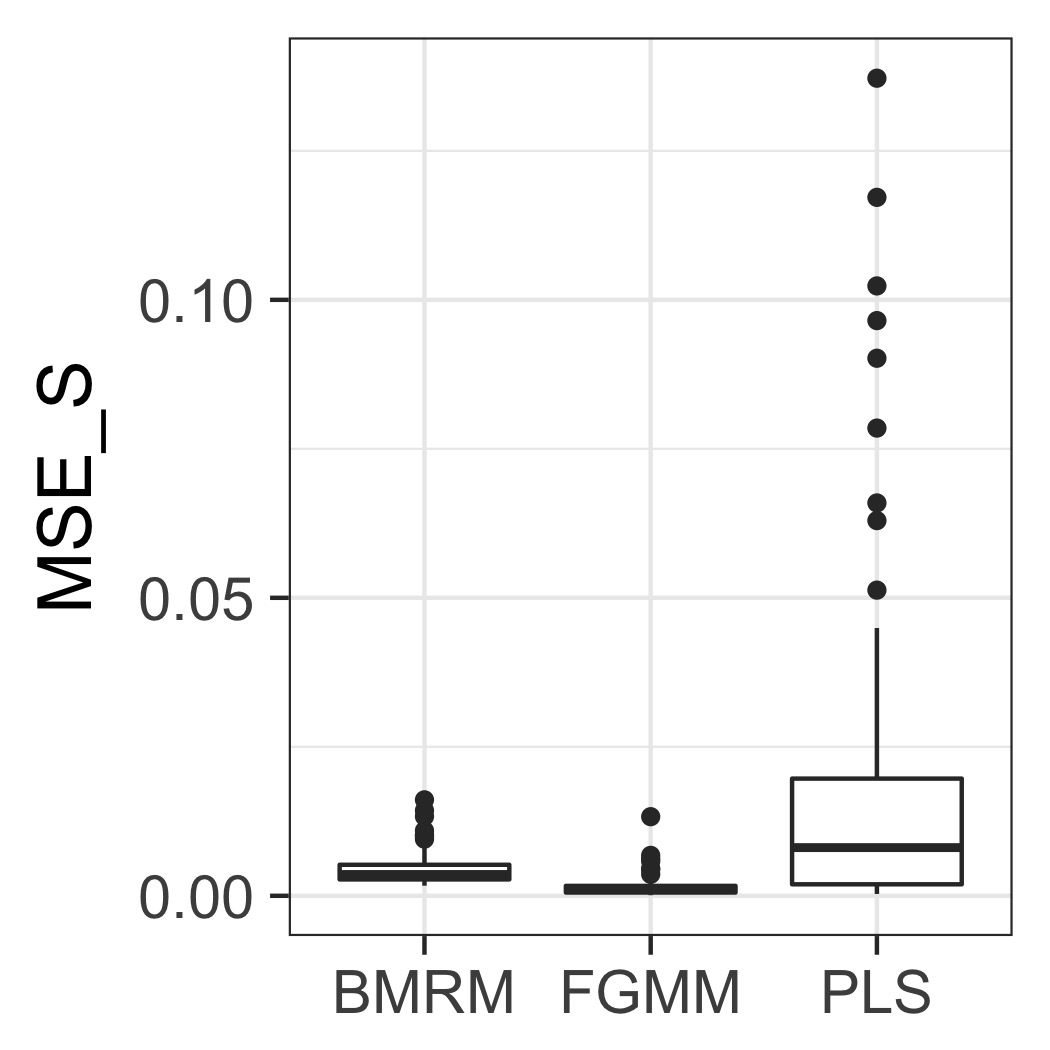}
\end{subfigure}\hfil 
\begin{subfigure}{0.25\textwidth}
  \includegraphics[width=\linewidth]{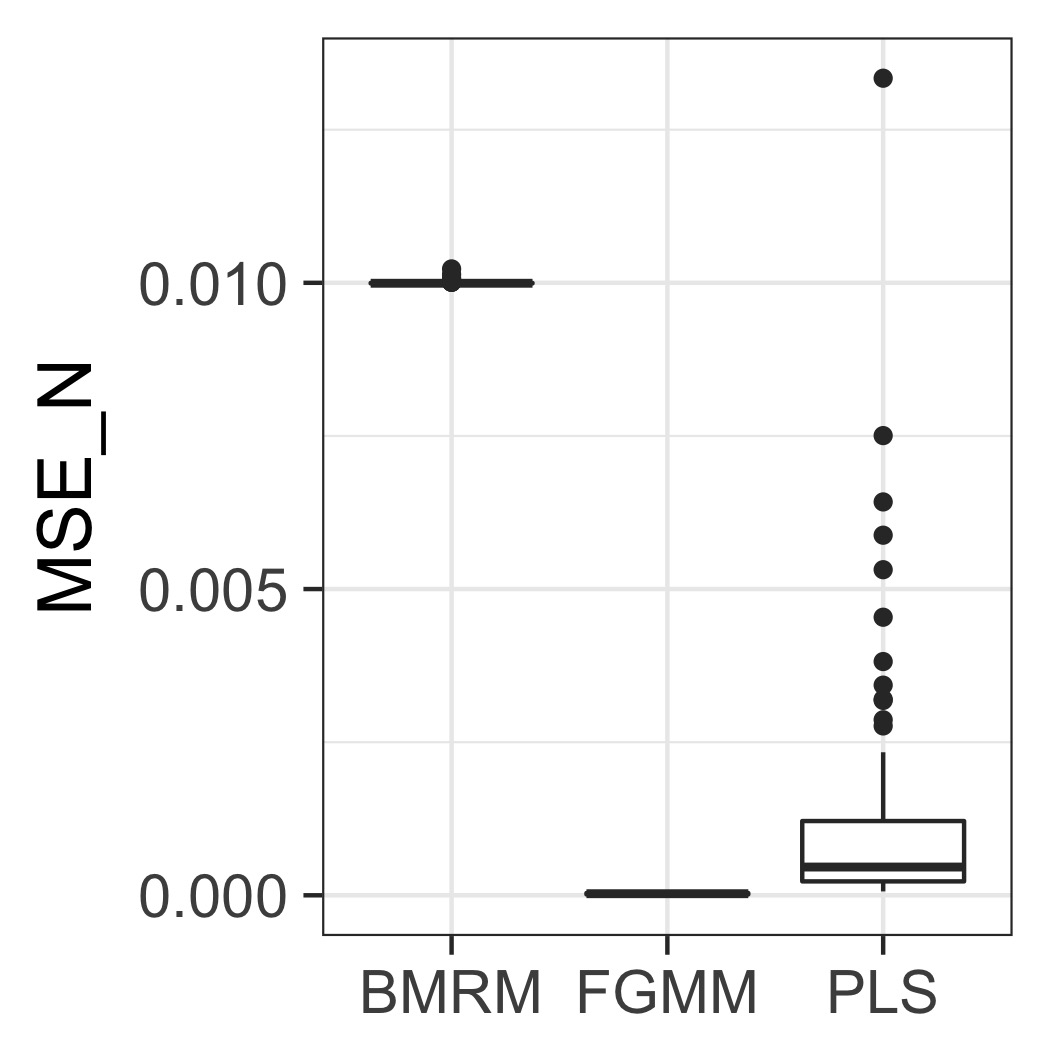}
\end{subfigure}

\medskip

\caption*{$(n,p) = (100,100)$, $m = 50$ and $\theta_0 = [1.25, -1.0, 1.75, -0.5, 0.375, {\bm{0}}]'$}
\begin{subfigure}{0.25\textwidth}
  \includegraphics[width=\linewidth]{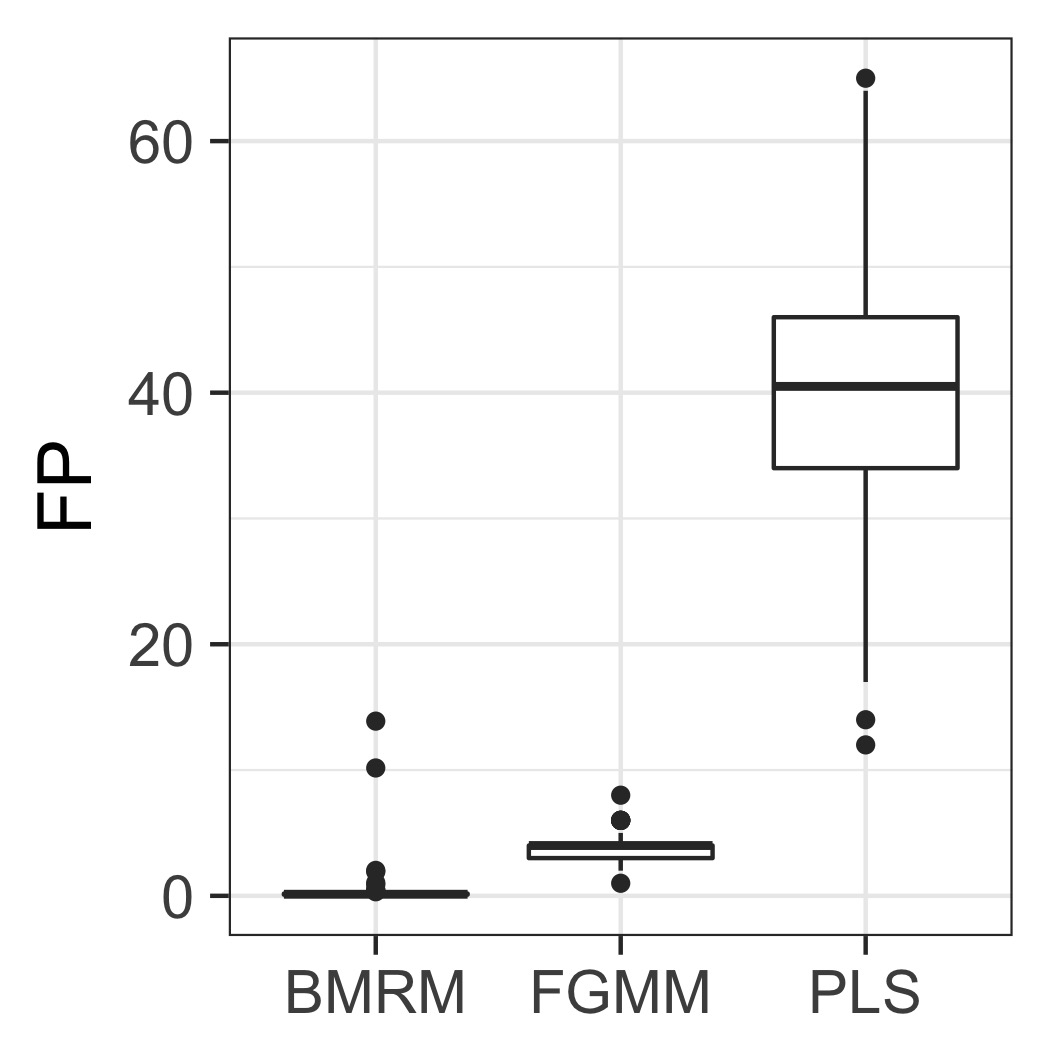}
\end{subfigure}\hfil 
\begin{subfigure}{0.25\textwidth}
  \includegraphics[width=\linewidth]{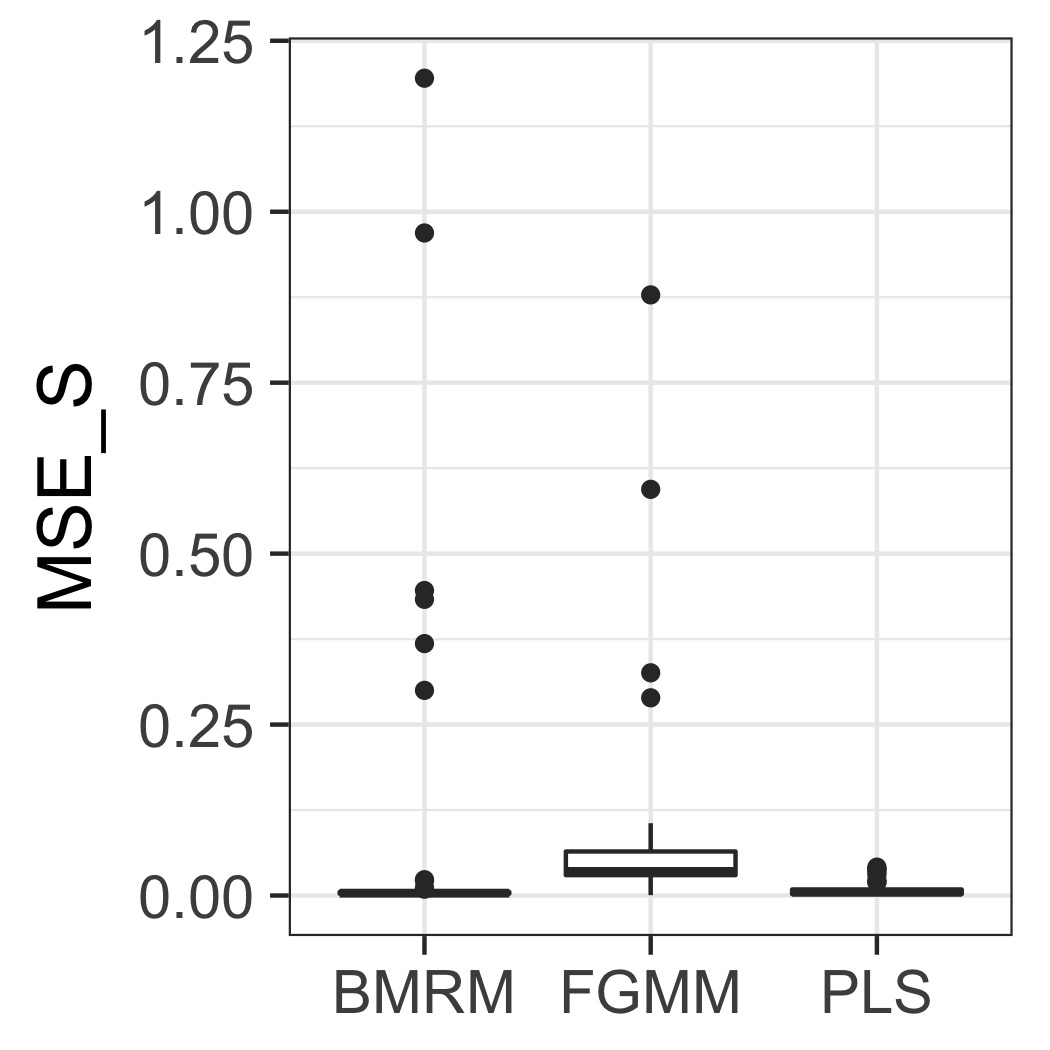}
\end{subfigure}\hfil 
\begin{subfigure}{0.25\textwidth}
  \includegraphics[width=\linewidth]{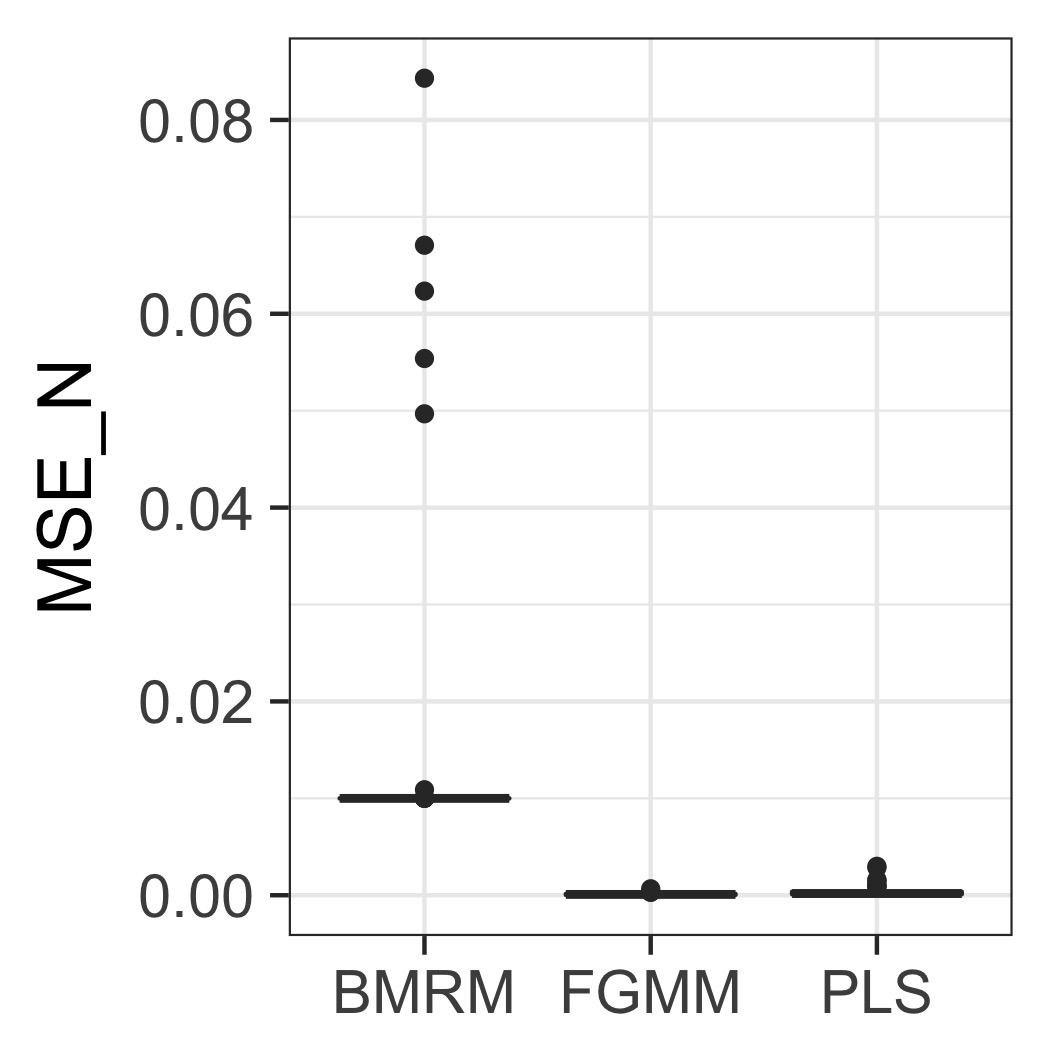}
\end{subfigure}

\medskip
\caption*{$(n,p) = (100,200)$, $m = 50$ and $\theta_0 = [5,-4,7,-2,1.5,{\bm{0}}]'$}
\begin{subfigure}{0.25\textwidth}
  \includegraphics[width=\linewidth]{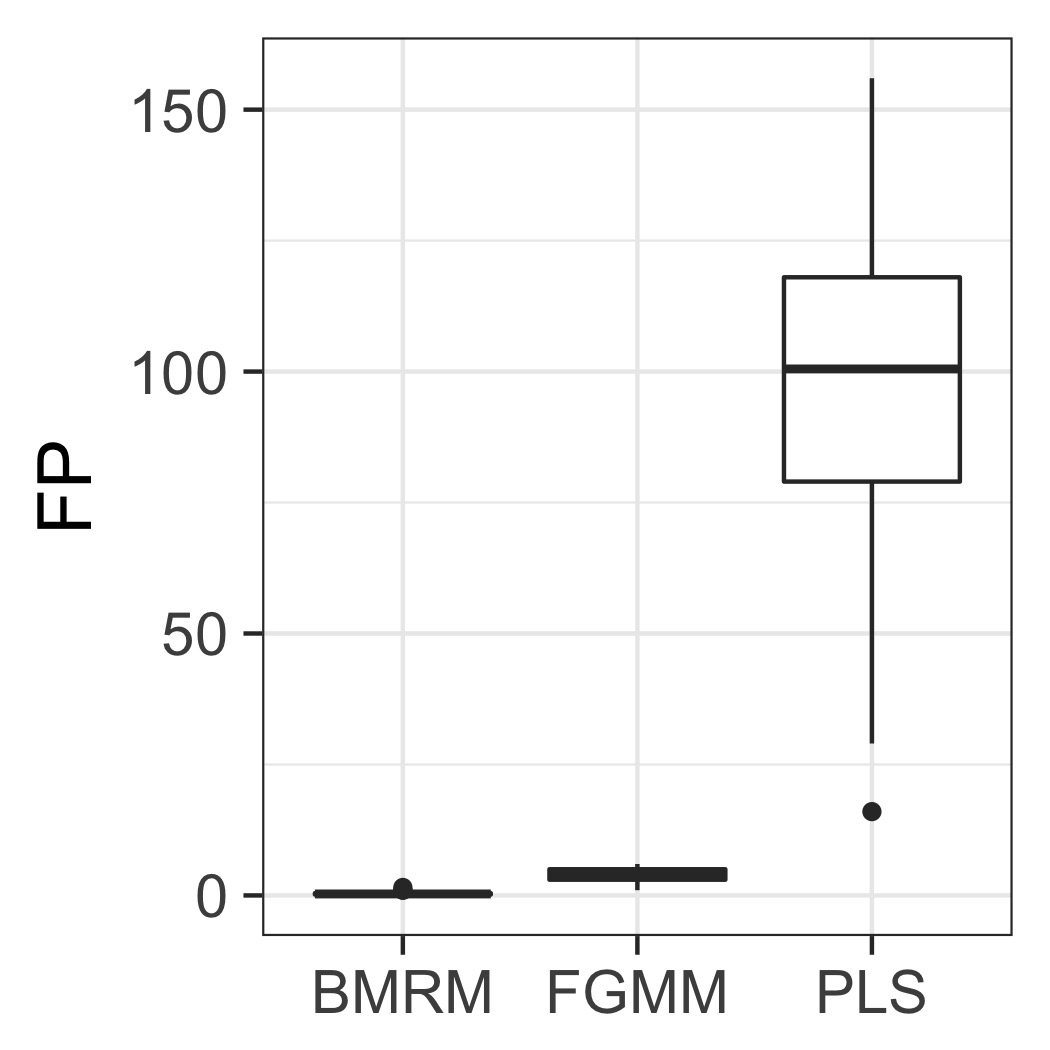}
\end{subfigure}\hfil 
\begin{subfigure}{0.25\textwidth}
  \includegraphics[width=\linewidth]{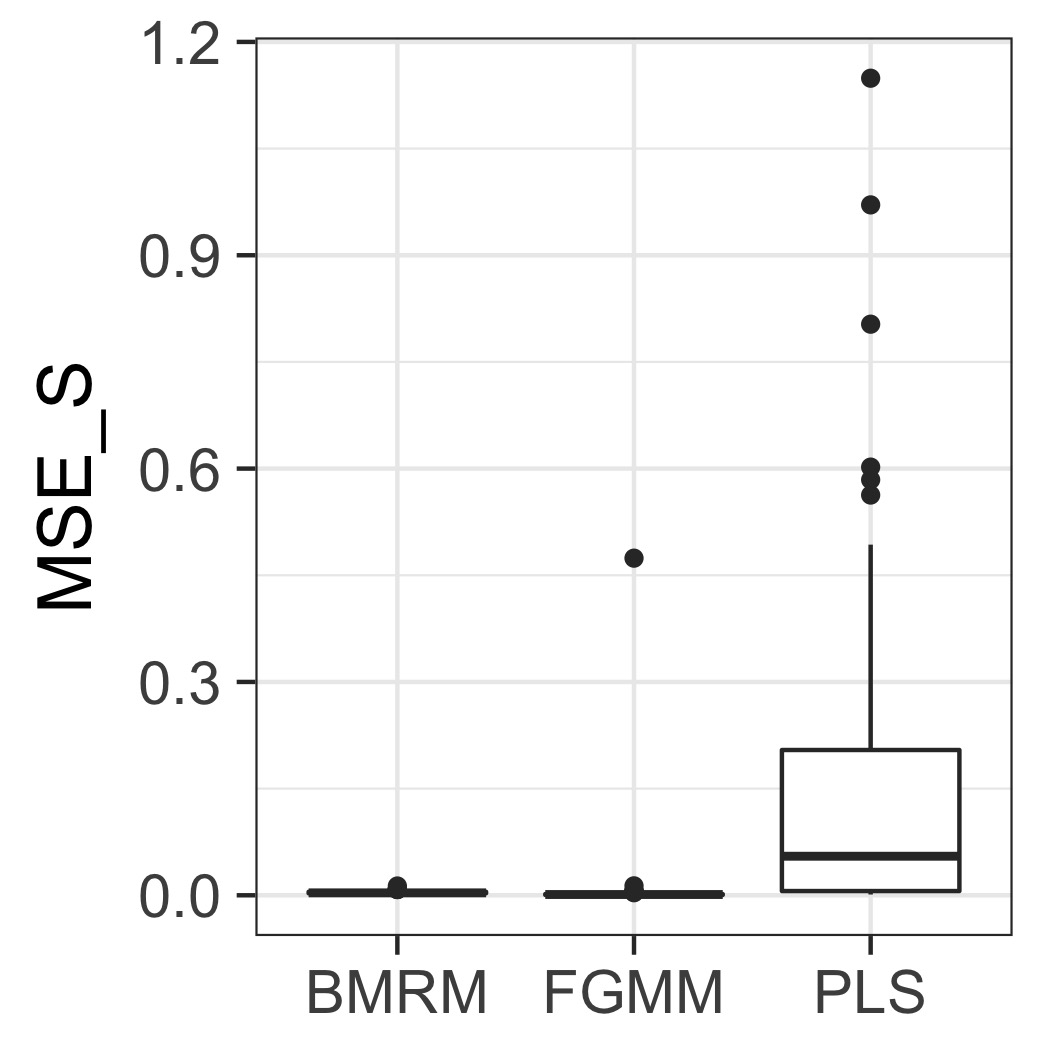}
\end{subfigure}\hfil 
\begin{subfigure}{0.25\textwidth}
  \includegraphics[width=\linewidth]{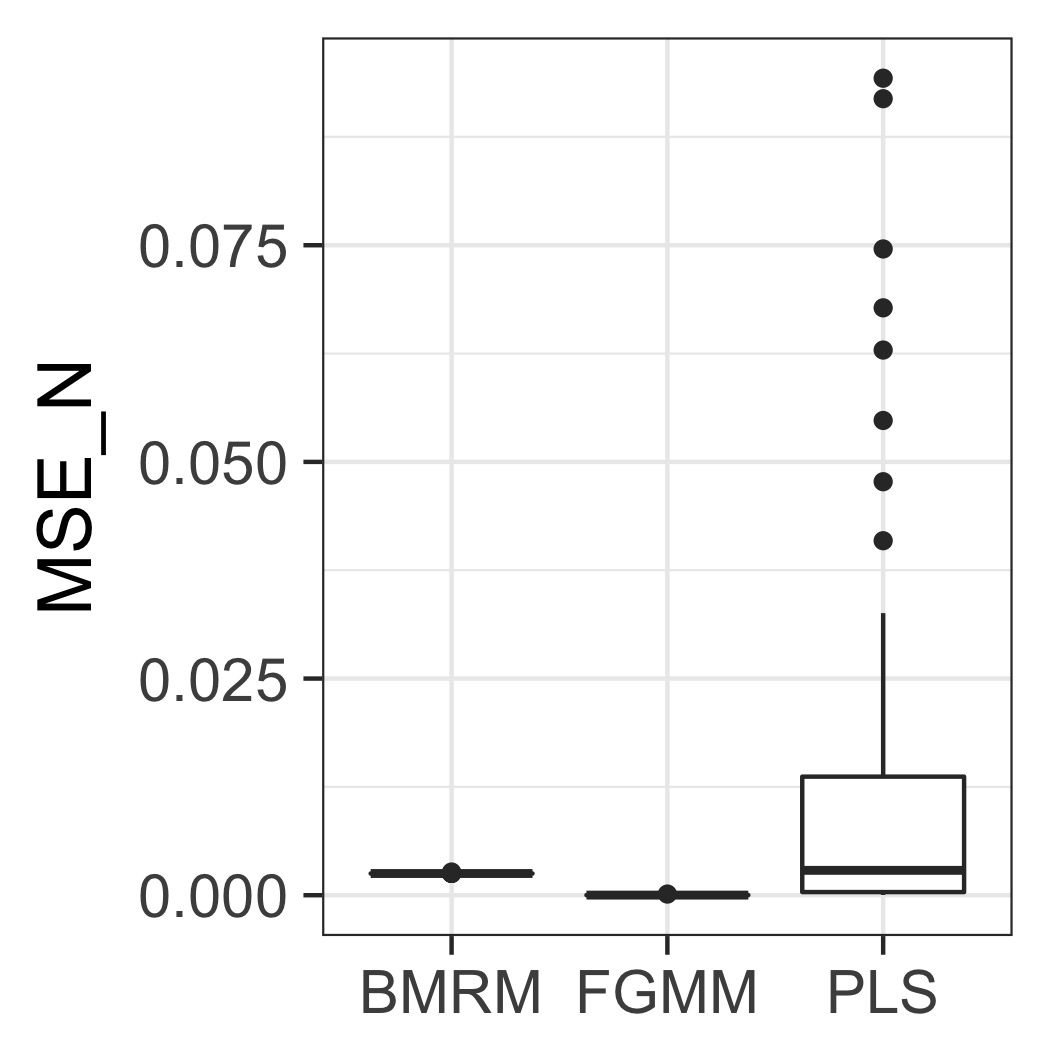}
\end{subfigure}

\medskip
\caption*{$(n,p) = (100,200)$, $m = 50$ and $\theta_0 = [1.25, -1.0, 1.75, -0.5, 0.375, {\bm{0}}]'$}
\begin{subfigure}{0.25\textwidth}
  \includegraphics[width=\linewidth]{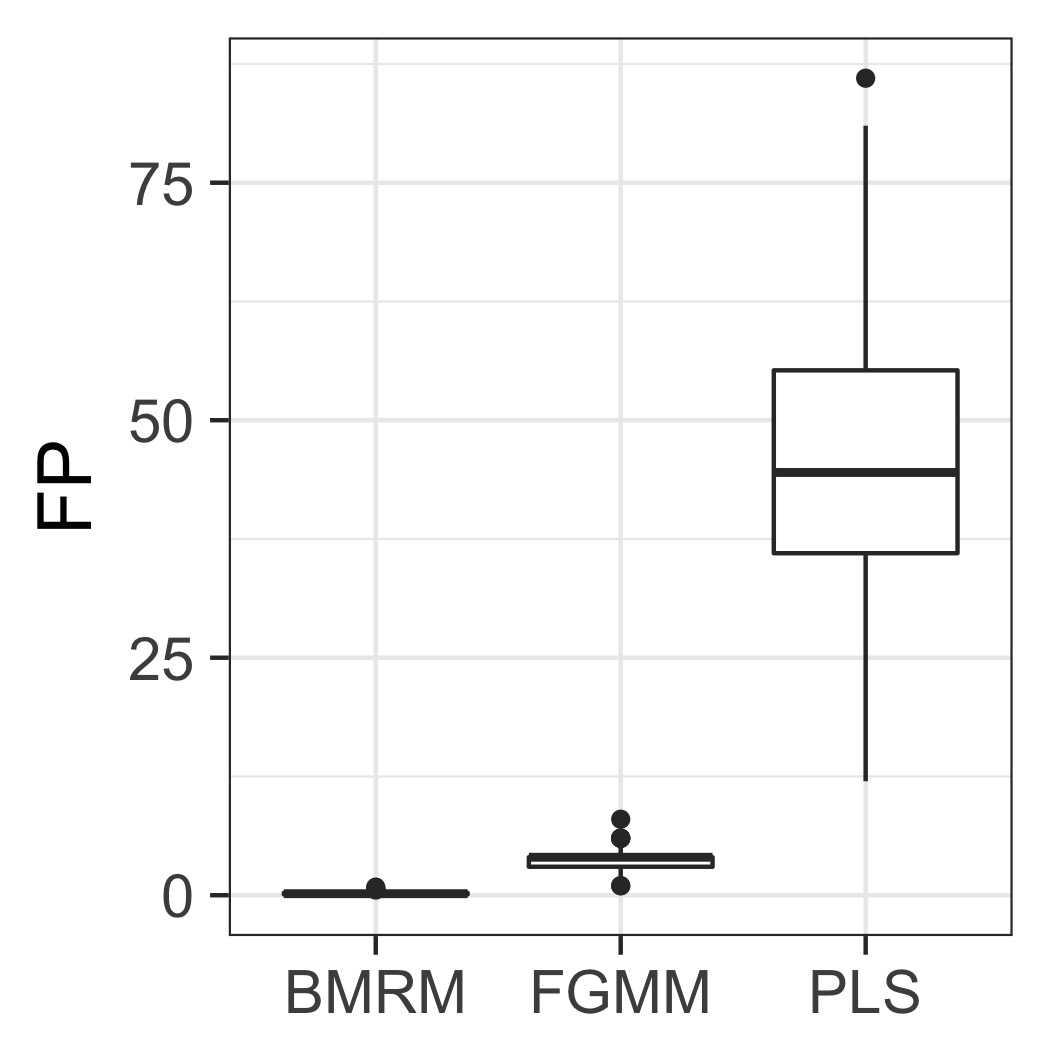}
\end{subfigure}\hfil 
\begin{subfigure}{0.25\textwidth}
  \includegraphics[width=\linewidth]{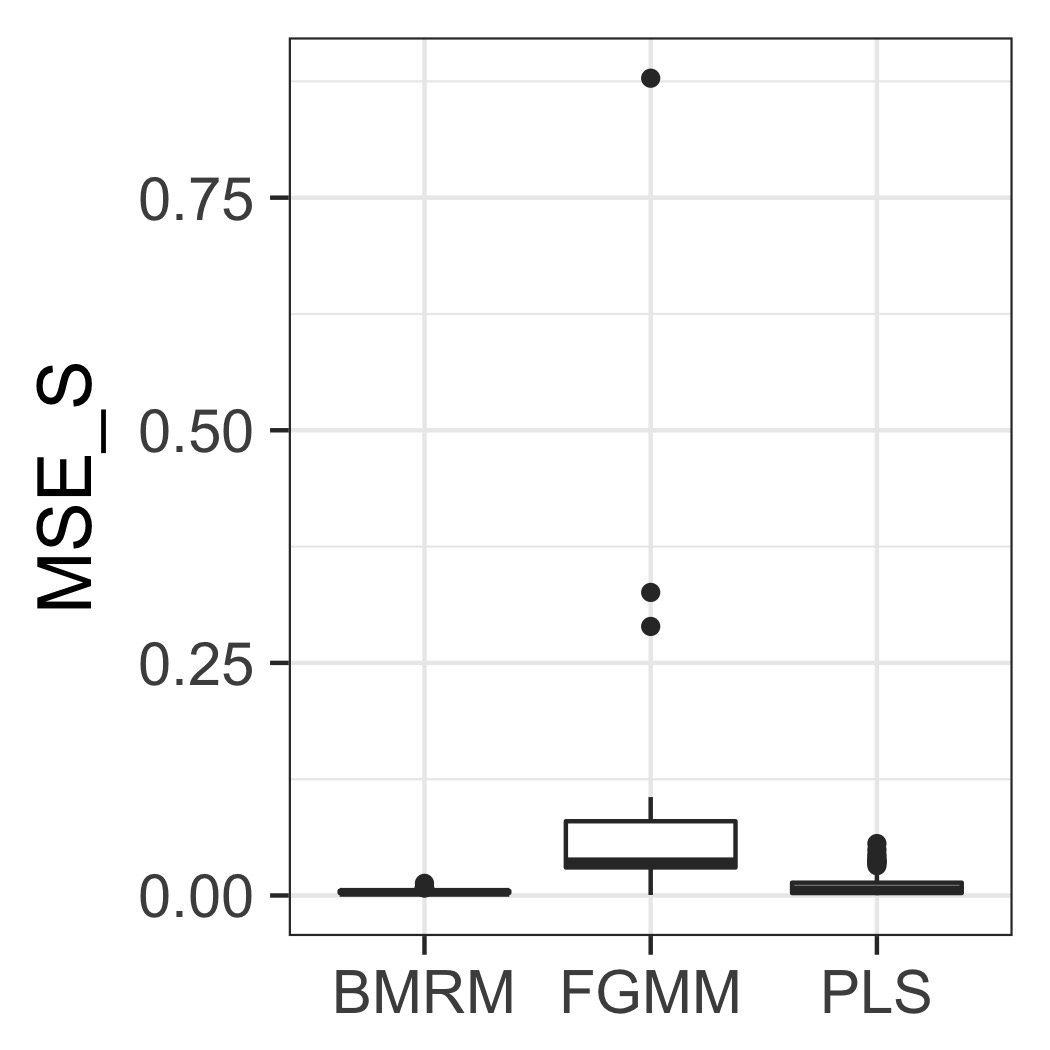}
\end{subfigure}\hfil 
\begin{subfigure}{0.25\textwidth}
  \includegraphics[width=\linewidth]{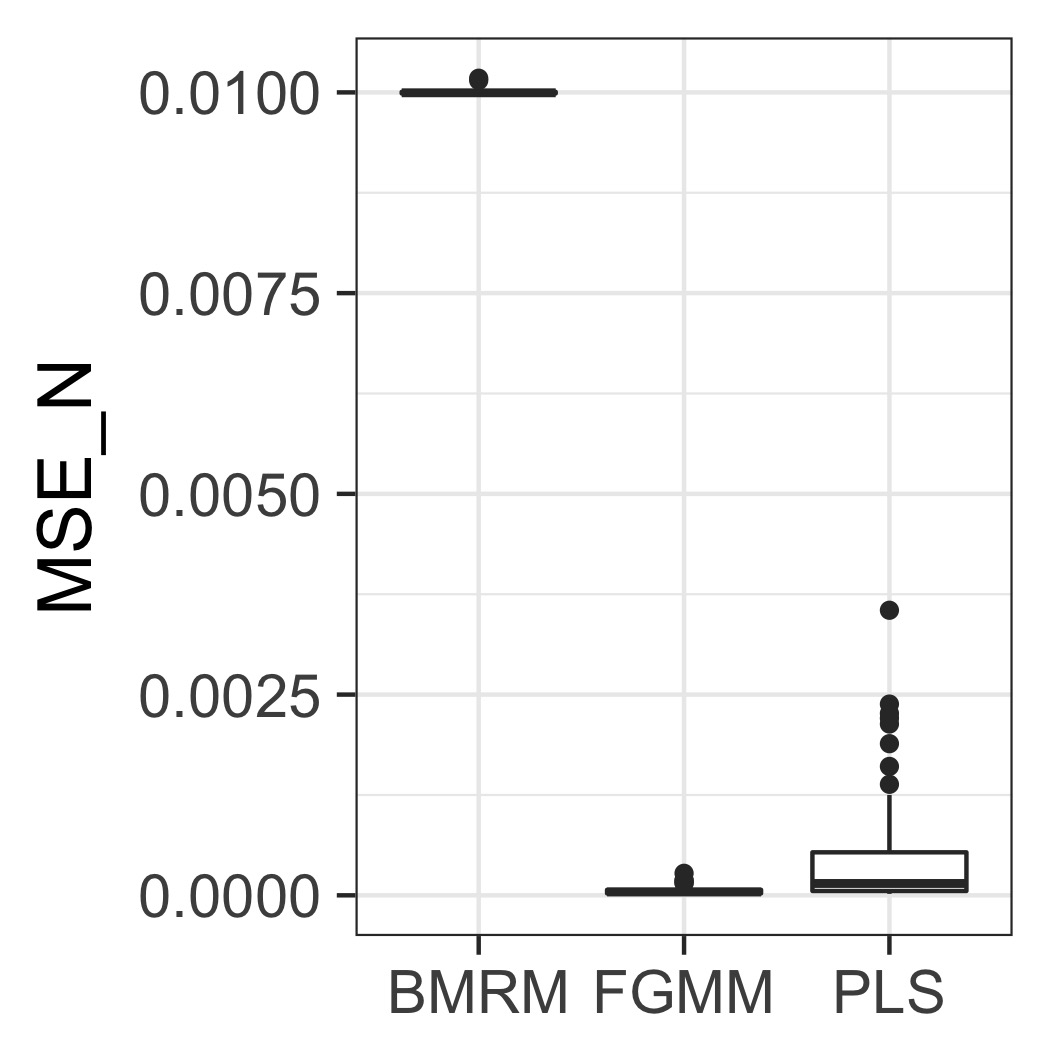}
\end{subfigure}

\caption{\small{Setup 1: False Positives (left), MSE for the active components (middle), and MSE for the inactive components (right) averaged over the runs of the MCMC sampler for 100 replicates.}}
\label{fig:setup1}
\end{figure}

\section{Endogeneity in Angrist \& Krueger Data} 
\citet{angrist1991does} use the large samples available in the U.S. Census to estimate wage equations where quarter of birth is used as an instrument for educational attainment. The coefficient of interest is $\theta_1$, which summarizes the causal impact of education on earning. We apply our method to the data that comes from the 1980 U.S. Census and consists of 329,509 males born in $1930 - 1939$. Consider the model, 
\begin{equation*}
y_i = \pscal{x_i}{\theta} + \epsilon_i, \quad \PE(\epsilon_i \vert w_i) = 0
\end{equation*}
where $y_i$ is the log(wage) of individual $i$ and $x_i$ denotes a set of 510 variables: education, 9 year-of-birth (YOB) dummies, 50 state-of-birth (SOB) dummies, and 450 state-of-birth $\times$ year-of-birth (YOB$\times$SOB) interactions. For individual $i$, we write 
\begin{equation*} 
x_i = [{\text{Education}}_{i}, {\text{YOB}}_i, {\text{SOB}}_i, ({{\text{YOB}}\times{\text{SOB}}})_i] \in {\mathbb{R}}^{510 \times 1}
\end{equation*} 

As instruments, $w_i$, we use 3 quarter-of-birth dummies (QOB) for the endogeneous variable education, and allow the exogeneous variables to be instruments for themselves. For individual $i$, we write 

\begin{equation*} 
w_i = [{\text{QOB}}_i, {\text{YOB}}_i, {\text{SOB}}_i, ({{\text{YOB}}\times{\text{SOB}}})_i] \in {\mathbb{R}}^{512 \times 1} 
\end{equation*} 

Note that there is an irregular dependence between the variables $x_i$ and their corresponding instruments $w_i$. For example, if the endogenous variable {\textit{education}} is active, then all 3 instruments, corresponding to $\mbox{QOB}$, are included in the model. 

\mbox{BMRM} selects a model with 9 covariates. The 95\% credible interval of $\theta_1$ is given by $[0.096, 0.129]$. 

\begin{figure}[H]
\centering 
\includegraphics[width=100mm]{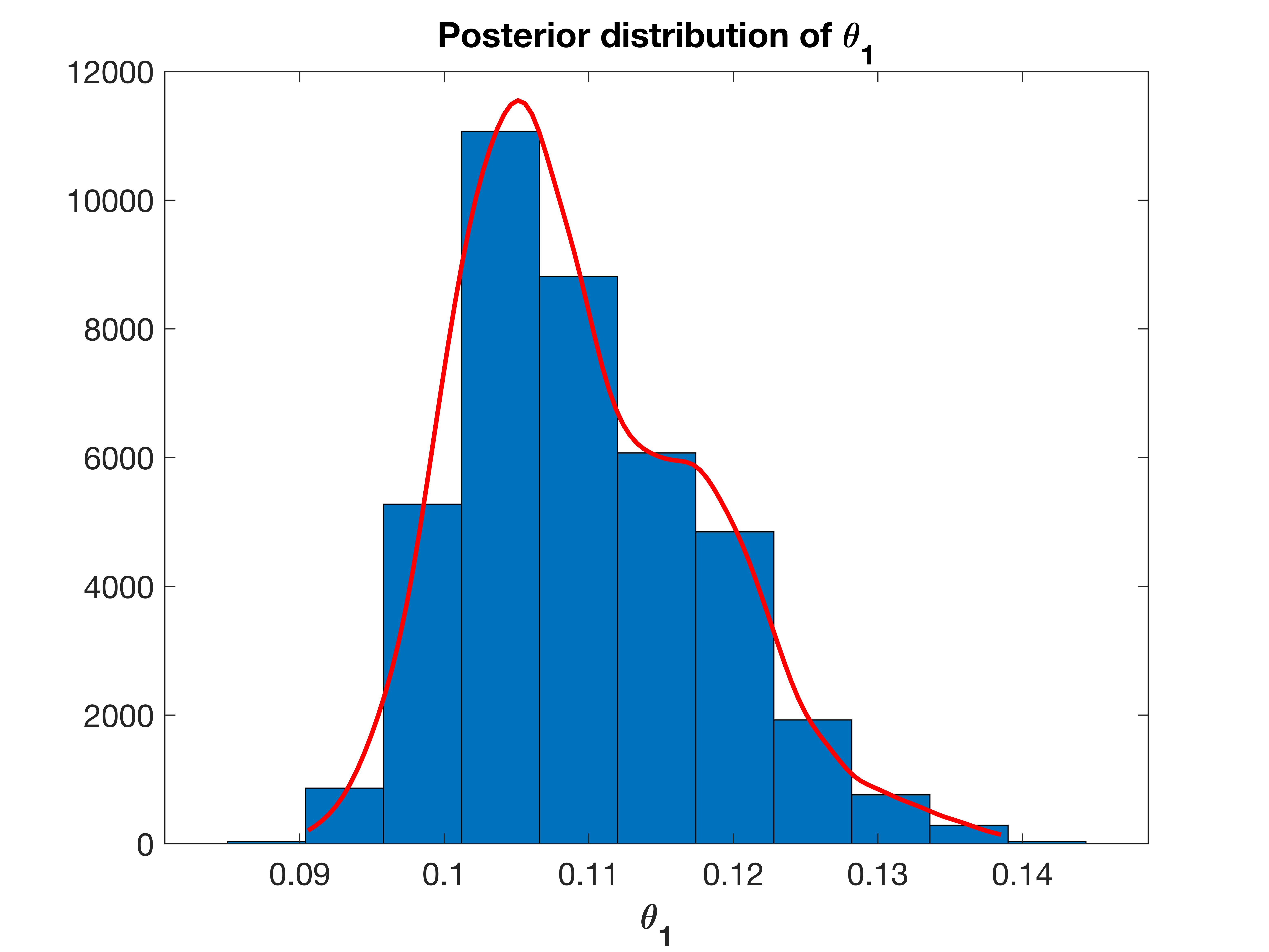}
\caption{Posterior distribution of $\theta_1$ which summarizes the causal impact of education on earning. The posterior mean is .1096.}
\label{theta_post} 
\end{figure}

\section{Proofs}\label{sec:proofs}
\subsection{Proof of Theorem \ref{thm1}}\label{sec:proof:thm1}
Our methods of proof are similar to techniques developed in \cite{castillo:etal:14,atchade:15b}.  For $\delta\in \Delta\eqdef\{0,1\}^p$, we will write $\mu_\delta(\rmd\theta)$ to denote the product measure on $\rset^p$ given by
\[\mu_{\delta}(\rmd \theta)\eqdef \prod_{j=1}^p\mu_{\delta_{j}}(\rmd \theta_{j}),\]
 where  $\mu_{0}(\rmd x)$ is the Dirac mass at $0$, and $\mu_1(\rmd x)$ is the Lebesgue measure on $\rset$.  First we derive a lower bound on the normalizing constant.
\begin{lemma}\label{lem:control:nc}
Assume H\ref{H1}-H\ref{H2}. Let $\check \Cset_\gamma(z)$ denote the normalizing constant of $\check\Pi(\cdot\vert z)$. For $z\in\e$,
\begin{equation}\label{eq:control:nc:check:G}
\Cset_{\gamma}(z) \geq   \omega_{\delta_\star} q_{\delta_\star,\theta_\star}(z)e^{-\frac{\rho^2}{2}\|\theta_\star\|_2^2} \left(\frac{\rho^2}{\frac{n\bar\kappa}{\lambda}+\rho^2}\right)^{
\frac{s_\star}{2}}.
\end{equation}
\end{lemma}
\begin{proof}
By definition we have
\begin{multline*}
\check \Cset_{\gamma}(z) =\sum_{\delta\in\Delta} \omega_\delta \int_{\rset^p} q_{\delta,\theta}(z)\frac{e^{-\frac{1}{2}\theta'B_\delta^{-1}\theta}}{\sqrt{\det(2\pi B_\delta)}}\rmd \theta \\
\geq \omega_{\delta_\star} q_{\delta_\star,\theta_\star}(z)\left(\frac{\rho^2}{2\pi}\right)^{\frac{s_\star}{2}}\int_{\rset^p} \frac{q_{\delta_\star,\theta}(z)}{q_{\delta_\star,\theta_\star}(z)} e^{-\frac{\rho^2}{2}\|\theta\|_2^2}\mu_{\delta_\star}(\theta).\end{multline*}
With $G(z) = \nabla\log q_{\delta_\star,\theta_\star}(z)$,  we have
\[\log q_{\delta_\star,\theta}(z) - \log q_{\delta_\star,\theta_\star}(z) = \pscal{G(z)}{\theta-\theta_\star} -\frac{1}{2\lambda}(\theta-\theta_\star)'X'W_{T(\delta_\star)}W_{T(\delta_\star)}'X(\theta-\theta_\star).\]
We recall that $M_\delta = W_{T(\delta)}'X$, so that for $z\in\e$,
\[\log q_{\delta_\star,\theta}(z) - \log q_{\delta_\star,\theta_\star}(z) \geq  \pscal{G(z)}{\theta-\theta_\star} -\frac{n\bar \kappa}{2\lambda}\|\theta-\theta_\star\|_2^2.\]
Hence,
\begin{multline*}
\check \Cset_\gamma(z) \geq \omega_{\delta_\star} q_{\delta_\star,\theta_\star}(z)\left(\frac{\rho^2}{2\pi}\right)^{\frac{s_\star}{2}}e^{-\frac{\rho^2}{2}\|\theta_\star\|_2^2}\\
\int_{\rset^p} e^{\pscal{G(z)}{\theta-\theta_\star}-\frac{\rho^2}{2}\left(\|\theta\|_2^2-\|\theta_\star\|_2^2\right)-\frac{n\bar\kappa}{2\lambda}\|\theta-\theta_\star\|_2^2} \mu_{\delta_\star}(\rmd \theta).\end{multline*}
We have $-\frac{\rho^2}{2}\left(\|\theta\|_2^2-\|\theta_\star\|_2^2\right) = -\frac{\rho^2}{2}\|\theta-\theta_\star\|_2^2 -\rho^2\pscal{\theta_\star}{\theta-\theta_\star}$. Therefore,
\begin{multline*}
\int_{\rset^p} e^{\pscal{G(z)}{\theta-\theta_\star}-\frac{\rho^2}{2}\left(\|\theta\|_2^2-\|\theta_\star\|_2^2\right)-\frac{n\bar\kappa}{2\lambda}\|\theta-\theta_\star\|_2^2} \mu_{\delta_\star}(\rmd \theta) \\
= \int_{\rset^p}e^{\pscal{G(z)-\rho^2\theta_\star}{u-\theta_\star}-\frac{\frac{n\bar\kappa}{\lambda}+\rho^2}{2}\|u-\theta_\star\|_2^2}\mu_{\delta_\star}(\rmd u) \geq \left(\frac{2\pi}{\frac{n\bar\kappa}{\lambda} + \rho^2}\right)^{\frac{s_\star}{2}},\end{multline*}
and (\ref{eq:control:nc:check:G}) also follows easily.
\end{proof}

Our proofs rely on the existence of some testing procedures that we take from \cite{A:B:2018}. 
Let $\Zset$ denote some sample space equipped with a reference sigma-finite measure. Let $f_\star$ be a density on $\Zset$. For each $\delta\in\Delta_{\bar s}$, suppose that we have $(\theta,z)\mapsto f_{\delta,\theta}(z)$ a jointly measurable $(0,+\infty)$-valued function on $\rset^p\times \Zset$ such that $\theta\mapsto \log f_{\delta,\theta}(z)$ is continuously differentiable for all $\delta\in\Delta_{\bar s}$, and $z\in\Zset$, and we denote its gradient by $\nabla \log f_{\delta,\theta}(z)\in\rset^p$. Given $\underline{\kappa}>0$, $\bar\rho>0$ and given $\theta_\star\in\rset^p$, we define
\begin{multline*}
\e_{\textsf{t}}\eqdef\left\{z\in\Zset:\;\sup_{\delta\in\Delta_{\bar s}}\|\nabla \log f_{\delta,\theta_\star}(z)\|_\infty\leq \frac{\bar\rho}{2},\;\; \mbox{ and for all } \delta\in\Delta_{\bar s},\;\theta\in\rset^p_\delta,\;\right.\\
\left.  \log f_{\delta,\theta}(z) -\log f_{\delta,\theta_\star}(z) -\pscal{\nabla \log f_{\delta,\theta_\star}(z)}{\theta-\theta_\star} \leq -\frac{\underline{\kappa}}{2}\|\theta-\theta_\star\|_2^2\right\}.
\end{multline*}
\begin{lemma}\label{test}
With the notations above, set $s_\star\eqdef\|\theta_\star\|_0$, $\epsilon \eqdef \frac{2(\bar s+s_\star)^{1/2}\bar\rho}{\underline{\kappa}}$. Then for any $M>2$, there exists a measurable function $\phi:\Zset\to [0,1]$ such that
\[\int_{\Zset} \phi(z)f_\star(z)\rmd z \leq 2(9p)^{\bar s} \frac{e^{-\frac{\underline{\kappa}}{32}(M\epsilon)^2}}{1-e^{-\frac{\underline{\kappa}}{32}(M \epsilon)^2}}.\]
Furthermore, for all $\delta\in\Delta_{\bar s}$ and all $\theta\in\rset^p_\delta$ such that  $\|\theta-\theta_\star\|_2 >jM\epsilon$ for some $j\geq 1$, we have
\[\int_{\e_{\textsf{t}}}(1-\phi(z)) \frac{f_{\delta,\theta}(z)}{f_{\delta,\theta_\star}(z)}f_{\star}(z)\rmd z \leq  e^{-\frac{\underline{\kappa}}{32}(jM\epsilon)^2}.\]
\end{lemma}
\begin{proof}See \cite{A:B:2018}, Lemma 11.
\end{proof}

\subsubsection{Proof of Theorem \ref{thm1}}
We have $\Delta \times \rset^p =  ((\Delta\setminus\Delta_{\bar s})\times\rset^p) \cup \bar\F_1 \cup \bar\F_{2}\cup \bar{\cB}_{m,M}$, where
\[
\bar\F_1\eqdef \bigcup_{\delta\in \Delta_{\bar s}} \{\delta\}\times \F_1^{(\delta)},\;\;\;\;\; \bar\F_{2}\eqdef \bigcup_{\delta\in\Delta_{\bar s}} \{\delta\}\times\F_2^{(\delta)},\]
where $\F_1^{(\delta)}\eqdef \left\{\theta\in\rset^p:\; \|\theta_\delta-\theta_\star\|_2> M\epsilon\right\}$, and \\
$\F_2^{(\delta)}\eqdef \left\{\theta\in\rset^p:\;\|\theta_\delta-\theta_\star\|_2\leq M\epsilon,\;\mbox{ and } \|\theta-\theta_\delta\|_2>m\sqrt{\gamma p}\right\}$. 
Since $\check\Pi$ is supported by $\Delta_{\bar s}$,  we have 
\[1 - \check\Pi(\bar{\cB}_{m,M}\vert Z)  = \check\Pi(\bar\F_1\vert Z) + \check\Pi(\bar\F_{2}\vert Z).\]
Setting $\F^{(\delta)}_2 =\F_{21}^{(\delta)}\cap\F_{22}^{(\delta)}$, where $\F^{(\delta)}_{21} =\{\theta\in\rset^p:\;\|\theta_\delta-\theta_\star\|_2\leq M\epsilon\}$,  and $\F^{(\delta)}_2 =\{\theta\in\rset^p:\;\|\theta-\theta_\delta\|_2 > m\sqrt{\gamma p}\}$, we have
\begin{multline}\check\Pi(\bar\F_2\vert z)  = \frac{\sum_{\delta\in\Delta} \omega_\delta \int_{\F_2^{(\delta)}} q_{\delta,\theta}(z)\frac{e^{-\frac{1}{2}\theta'B_\delta^{-1}\theta}}{\sqrt{\det(2\pi B_\delta)}}\rmd \theta}{\sum_{\delta\in\Delta} \omega_\delta \int_{\rset^p} q_{\delta,\theta}(z)\frac{e^{-\frac{1}{2}\theta'B_\delta^{-1}\theta}}{\sqrt{\det(2\pi B_\delta)}}\rmd \theta}\\
 = \frac{\sum_{\delta\in\Delta} \omega_\delta \left(\frac{\rho^2}{2\pi}\right)^{\frac{\|\delta\|_0}{2}}\left\{\int_{\F_{21}^{(\delta)}} q_{\delta,\theta}(z)e^{-\frac{\rho^2}{2}\|\theta\|_2^2}\mu_\delta(\rmd \theta)\right\}\left\{ \PP\left(V\in\F_{22}^{(\delta)}\right)\right\}}{\sum_{\delta\in\Delta} \omega_\delta \left(\frac{\rho^2}{2\pi}\right)^{\frac{\|\delta\|_0}{2}}\int_{\rset^p} q_{\delta,\theta}(z)e^{-\frac{\rho^2}{2}\|\theta\|_2^2}\mu_\delta(\rmd \theta)},\end{multline}
where $V\sim\textbf{N}_p(0,\gamma I_p)$. By standard Guassian deviation bound, it is easy to see that $\PP(V\in\F_{22}^{(\delta)}) \leq 2e^{-\frac{(m-1)^2p}{2}}$ for all $\delta\in\Delta_{\bar s}$. It follows that for all $z\in\Zset$, $\check\Pi(\bar\F_2\vert z) \leq 2e^{-\frac{(m-1)^2p}{2}}$.

Similarly, note that
\begin{eqnarray*}\check\Pi(\bar\F_1\vert z)  & = & \frac{\sum_{\delta\in\Delta} \omega_\delta \int_{\F_1^{(\delta)}} q_{\delta,\theta}(z)\frac{e^{-\frac{1}{2}\theta'B_\delta^{-1}\theta}}{\sqrt{\det(2\pi B_\delta)}}\rmd \theta}{\sum_{\delta\in\Delta} \omega_\delta \int_{\rset^p} q_{\delta,\theta}(z)\frac{e^{-\frac{1}{2}\theta'B_\delta^{-1}\theta}}{\sqrt{\det(2\pi B_\delta)}}\rmd \theta}\\
& = & \frac{\sum_{\delta\in\Delta} \omega_\delta \left(\frac{\rho^2}{2\pi}\right)^{\frac{\|\delta\|_0}{2}} \int_{\F_2^{(\delta)}} q_{\delta,\theta}(z)e^{-\frac{\rho^2}{2}\|\theta\|_2^2}\mu_\delta(\rmd \theta)}{\sum_{\delta\in\Delta} \omega_\delta \left(\frac{\rho^2}{2\pi}\right)^{\frac{\|\delta\|_0}{2}} \int_{\rset^p} q_{\delta,\theta}(z)e^{-\frac{\rho^2}{2}\|\theta\|_2^2}\mu_\delta(\rmd \theta)}.\end{eqnarray*}

We apply Lemma \ref{test} with $\theta_\star$ as in H\ref{H1}, $f_\star$ equal to the joint density of $z=(y,X,W)$ as assumed in H\ref{H1}, and $f_{\delta,\theta}(z) = q_{\delta,\theta}(z)$. In that case for $\delta\in\Delta_{\bar s}$, $\theta\in\rset^p_\delta$, we have
\begin{multline*}
\log f_{\delta,\theta}(z) - \log f_{\delta,\theta_\star}(z) -\pscal{\nabla \log f_{\delta,\theta_\star}(z)}{\theta-\theta_\star} = -\frac{1}{2\lambda}(\theta-\theta_\star)'(M_\delta'M_\delta)(\theta-\theta_\star)\\
 \leq -\frac{n\underline{\kappa}}{2\lambda}\|\theta-\theta_\star\|_2^2,\end{multline*}
for $z\in\e$. And $\nabla\log  f_{\delta,\theta_\star}(z) = \frac{1}{\lambda}M_\delta' W_{T(\delta)}' (y-X\theta_\star)$. It follows that the $j$-th component of $\nabla\log  f_{\delta,\theta_\star}(z)$ -- denoted $\nabla_j\log  f_{\delta,\theta_\star}(z)$ -- satisfies
\begin{multline*}
\left|\nabla_j\log  f_{\delta,\theta_\star}(z)\right| =\frac{1}{\lambda}\left|\sum_{i:\;T(\delta)_i\neq 0} M_{\delta,ij}\pscal{W_{T(\delta),i}}{y-X\theta_\star}\right| \\
\leq \frac{1}{\lambda} \sup_{1\leq k\leq q} \left|\pscal{W_k}{y-X\theta_\star}\right|\sum_{i:\;T(\delta)_i\neq 0}|M_{\delta,ij}|\\
\leq \sigma_0\frac{\kappa(1)}{\lambda} \sqrt{2n\bar t\log(pq)},
\end{multline*}
for $z\in\e$, where we recall that $\bar t = \max_{\delta\in\Delta_{\bar s}} \|T(\delta)\|_0$. Hence we can apply Lemma \ref{test} with $\underline{\kappa}$ taken as $\frac{n\underline{\kappa}}{\lambda}$ and $\bar\rho$ taken as $2\sigma_0\frac{\kappa(1)}{\lambda}\sqrt{2n\bar t\log(pq)}$. In that case we have 
\[\epsilon = \frac{(\bar s+s_\star)^{1/2}\bar \rho}{\underline{\kappa}} =2\sqrt{2}\sigma_0 \frac{\bar\kappa(1)}{\underline{\kappa}} \sqrt{\frac{(\bar s + s_\star)\bar t\log(pq)}{n}}. \]

 Let $\phi$ denote the test function asserted by Lemma \ref{test} below, where $M>2$ is some arbitrary absolute constant. We can then write 
\begin{equation}\label{eq:proof:thm:contrac:eq2}
\PE_\star\left[\textbf{1}_{\e}(z)\check\Pi(\bar\F_1\vert z)\right] \leq \PE_\star\left(\phi(z)\right)+ \PE_\star\left[\textbf{1}_{\e}(z)\left(1-\phi(z)\right) \check\Pi(\bar\F_1\vert z)\right].\end{equation}
Lemma \ref{test} gives
\begin{equation}\label{eq:proof:thm:contrac:eq3}
\PE_\star\left(\phi(z)\right) \leq 2(9p)^{(\bar s)}\frac{e^{-\frac{\underline{\kappa}}{32}(M\epsilon)^2}}{1-e^{-\frac{\underline{\kappa}}{32}(M\epsilon)^2}}\leq \frac{1}{p^{M^2(1+s_\star)}},\end{equation}
for all $p$ large enough. By Lemma \ref{lem:control:nc},  we have
\begin{multline*}
\check\Pi(\bar\F_1\vert z)\textbf{1}_{\e}(Z)  \leq  \left(1 +\frac{\bar\kappa(s_\star)}{\rho^2}\right)^{\frac{s_\star}{2}} \\
\times\textbf{1}_{\e}(z)  \sum_{\delta\in\Delta_{\bar s}}\frac{\omega_\delta}{\omega_{\delta_\star}}\left(\frac{\rho^2}{2\pi }\right)^{\frac{\|\delta\|_0}{2}}\int_{\F_1}\frac{q_{\delta,\theta}(z)}{q_{\delta_\star,\theta_\star}(z)} e^{-\frac{\rho^2}{2}\left(\|\theta\|_2^2-\|\theta_\star\|_2^2\right)}\mu_\delta(\rmd \theta),\end{multline*}
where $\F_1\eqdef\{\theta\in\rset^p:\; \|\theta-\theta_\star\|_2\leq M\epsilon\}$.   We have
\[\frac{q_{\delta,\theta_\star}(z)}{q_{\delta_\star,\theta_\star}(z)} = \exp\left(\frac{1}{2\lambda}\epsilon\left[W_{T(\delta_\star)}W_{T(\delta_\star)}' - W_{T(\delta)}W_{T(\delta)}'\right]\epsilon\right) \leq \exp\left(\frac{1}{2\lambda}\epsilon\left[W_{T(\delta_\star)}W_{T(\delta_\star)}'\right]\epsilon\right),\]
and for $z\in\e$, $\epsilon\left[W_{T(\delta_\star)}W_{T(\delta_\star)}'\right]\epsilon\leq \bar t\sigma_0^2\log(pq)$.
It follows from the above and Fubini's theorem that
\begin{multline}\label{eq:proof:thm:contrac:eq1}
\PE_\star\left[\textbf{1}_{\e}(z)(1-\phi(z))\Pi( \bar\F_1\vert z)\right]  \leq  e^{\frac{\sigma_0^2\bar t}{\lambda}\log(pq)}\left(1 +\frac{\bar\kappa(s_\star)}{\rho^2}\right)^{\frac{s_\star}{2}} \\
 \times \sum_{\delta\in\Delta_{\bar s}}\frac{\omega_\delta}{\omega_{\delta_\star}}\left(\frac{\rho^2}{2\pi}\right)^{\frac{\|\delta\|_0}{2}}\int_{\F_1}\PE_\star\left[\textbf{1}_{\e}(z)\left(1-\phi(z)\right) \frac{q_{\delta,\theta}(z)}{q_{\delta,\theta_\star}(z)}\right] e^{-\frac{\rho^2}{2}\left(\|\theta\|_2^2-\|\theta_\star\|_2^2\right)}\mu_\delta(\rmd \theta),\end{multline}
We write  $\F_1 = \cup_{j\geq 1} \F_{1,j}$, where $\F_{1,j}\eqdef\{\theta\in\rset^p:\; jM\epsilon < \|\theta-\theta_\star\|_2\leq (j+1)M\epsilon\}$. Using this and Lemma \ref{test}, we have
\begin{multline*}
\int_{\F_{1,j}}\PE_\star\left[\textbf{1}_{\e}(z)\left(1-\phi(z)\right) \frac{q_{\delta,\theta}(z)}{q_{\delta,\theta_\star}(z)}\right] e^{-\frac{\rho^2}{2}\left(\|\theta\|_2^2-\|\theta_\star\|_2^2\right)}\mu_\delta(\rmd \theta)\\
\leq e^{-\frac{\underline{\kappa}}{32}(jM\epsilon)^2}\int_{\F_{1,j}} e^{-\frac{\rho^2}{2}\left(\|\theta\|_2^2-\|\theta_\star\|_2^2\right)}\mu_\delta(\rmd \theta),
\end{multline*}
and
\begin{multline*}
 \int_{\F_{1,j}}e^{-\frac{\rho^2}{2}(\|\theta\|_2^2 - \|\theta_\star\|_2^2)}\mu_\delta(\rmd \theta) =  \int_{\F_{1,j}}e^{-\frac{\rho^2}{2}\left(\|\theta-\theta_\star\|_2^2 + 2\pscal{\theta_\star}{\theta-\theta_\star}\right)}\mu_\delta(\rmd \theta)\\
  \leq e^{2\rho^2\|\theta_\star\|_2(jM\epsilon)} \int_{\rset^p}e^{-\frac{\rho^2}{2}\|\theta-\theta_\star\|_2^2}\mu_\delta(\rmd\theta)\leq e^{2\rho^2\|\theta_\star\|_2(jM\epsilon)}\left(\frac{2\pi}{\rho^2}\right)^{\frac{\|\delta\|_2}{2}}. \end{multline*}
Therefore (\ref{eq:proof:thm:contrac:eq1}) becomes
\begin{multline}\label{eq:proof:thm:contrac:eq2}
\PE_\star\left[\textbf{1}_{\e_{\bar\rho}}(z)(1-\phi(z))\Pi( \bar\F_1\vert z)\right] \\
 \leq  (pq)^{\frac{\sigma_0^2\bar t}{\lambda}}\left(1 +\frac{\bar\kappa(s_\star)}{\rho^2}\right)^{\frac{s_\star}{2}} \sum_{\delta\in\Delta_{\bar s}}\frac{\omega_\delta}{\omega_{\delta_\star}}\sum_{j\geq 1}e^{-\frac{\underline{\kappa}}{32}(jM\epsilon)^2 + 2\rho^2\|\theta_\star\|_2(jM\epsilon)}\\
\leq  (pq)^{\frac{\sigma_0^2\bar t}{\lambda}}\left(1 +\frac{\bar\kappa(s_\star)}{\rho^2}\right)^{\frac{s_\star}{2}} \sum_{\delta\in\Delta_{\bar s}}\frac{\omega_\delta}{\omega_{\delta_\star}} \frac{e^{-\frac{\underline{\kappa}}{64}(M\epsilon)^2}}{1-e^{-\frac{\underline{\kappa}}{64}(M\epsilon)^2}},
 \end{multline}
 where we use the fact that for $M>128$, since $\rho^2\|\theta_\star\|_\infty\leq \bar\rho$, we have
\[-\frac{\underline{\kappa}}{64}(jM\epsilon)^2 + 2\rho^2\|\theta_\star\|_2(jM\epsilon) \leq 0.\]
We note that for $\q\leq 1/2$, and since ${p\choose s}\leq p^s$,
\begin{multline*}
\sum_{\delta\in\Delta_{\bar s}}\frac{\omega_\delta}{\omega_{\delta_\star}} = \left(\frac{1-\q}{\q}\right)^{s_\star}\sum_{\delta\in\Delta_{\bar s}}\left(\frac{\q}{1-\q}\right)^{\|\delta\|_0}\leq \left(\frac{1-\q}{\q}\right)^{s_\star} \sum_{s=0}^{\bar s}{p\choose s}(2\q)^{s} \\
\leq p^{s_\star(1+u)}\sum_{s=0}^{\bar s}(2p\q)^{s} \leq 2p^{s_\star(1+u)},
\end{multline*} 
provided that $p^u\geq 4$. It follows readily that for all $p$ large enough, and $M>u$,
\begin{equation}\label{eq:proof:thm:contrac:eq42}
\PE_\star\left[\textbf{1}_{\e_{\bar\rho}}(Z)(1-\phi(Z))\Pi( \F_1\vert Z)\right] \leq   \frac{(pq)^{\frac{\sigma_0^2\bar t}{\lambda}}}{p^{M^2(1+s_\star)}}.
  \end{equation}
The result follows by putting the pieces together. 
\vspace{-0.3cm}
\begin{flushright}
$\square$
\end{flushright}
\medskip

\bibliographystyle{ims}
\bibliography{biblio_graph,biblio_mcmc,biblio_optim,biblio_notes}

\begin{thebibliography}{22}
\expandafter\ifx\csname natexlab\endcsname\relax\def\natexlab#1{#1}\fi
\expandafter\ifx\csname url\endcsname\relax
  \def\url#1{\texttt{#1}}\fi
\expandafter\ifx\csname urlprefix\endcsname\relax\def\urlprefix{URL }\fi

\bibitem[{Angrist and Keueger(1991)}]{angrist1991does}
\textsc{Angrist, J.~D.} and \textsc{Keueger, A.~B.} (1991).
\newblock Does compulsory school attendance affect schooling and earnings?
\newblock \textit{The Quarterly Journal of Economics} \textbf{106} 979--1014.

\bibitem[{Atchade(2017)}]{atchade:15b}
\textsc{Atchade, Y.~A.} (2017).
\newblock On the contraction properties of some high-dimensional
  quasi-posterior distributions.
\newblock \textit{Ann. Statist.} \textbf{45} 2248--2273.

\bibitem[{Atchad{\'e} et~al.(2017)}]{atchade2017contraction}
\textsc{Atchad{\'e}, Y.~A.} \textsc{et~al.} (2017).
\newblock On the contraction properties of some high-dimensional
  quasi-posterior distributions.
\newblock \textit{The Annals of Statistics} \textbf{45} 2248--2273.

\bibitem[{{Banerjee} and {Ghosal}(2013)}]{banerjee:ghosal13}
\textsc{{Banerjee}, S.} and \textsc{{Ghosal}, S.} (2013).
\newblock {Posterior convergence rates for estimating large precision matrices
  using graphical models}.
\newblock \textit{ArXiv e-prints} .

\bibitem[{Belloni et~al.(2017)Belloni, Chernozhukov, Hansen and
  Newey}]{belloni2017simultaneous}
\textsc{Belloni, A.}, \textsc{Chernozhukov, V.}, \textsc{Hansen, C.} and
  \textsc{Newey, W.} (2017).
\newblock Simultaneous confidence intervals for high-dimensional linear models
  with many endogenous variables.
\newblock \textit{arXiv preprint arXiv:1712.08102} .

\bibitem[{B{\"u}hlmann and van~de Geer(2011)}]{buhlGeer11}
\textsc{B{\"u}hlmann, P.} and \textsc{van~de Geer, S.} (2011).
\newblock \textit{Statistics for high-dimensional data}.
\newblock Springer Series in Statistics, Springer, Heidelberg.
\newblock Methods, theory and applications.

\bibitem[{Candes et~al.(2007)Candes, Tao et~al.}]{candes2007dantzig}
\textsc{Candes, E.}, \textsc{Tao, T.} \textsc{et~al.} (2007).
\newblock The dantzig selector: Statistical estimation when p is much larger
  than n.
\newblock \textit{The Annals of Statistics} \textbf{35} 2313--2351.

\bibitem[{Castillo et~al.(2015)Castillo, Schmidt-Hieber and van~der
  Vaart}]{castillo:etal:14}
\textsc{Castillo, I.}, \textsc{Schmidt-Hieber, J.} and \textsc{van~der Vaart,
  A.} (2015).
\newblock Bayesian linear regression with sparse priors.
\newblock \textit{Ann. Statist.} \textbf{43} 1986--2018.

\bibitem[{Castillo and van~der Vaart(2012)}]{castillo:etal:12}
\textsc{Castillo, I.} and \textsc{van~der Vaart, A.} (2012).
\newblock Needles and straw in a haystack: Posterior concentration for possibly
  sparse sequences.
\newblock \textit{Ann. Statist.} \textbf{40} 2069--2101.

\bibitem[{Chernozhukov and Hong(2003)}]{chernozhukov2003mcmc}
\textsc{Chernozhukov, V.} and \textsc{Hong, H.} (2003).
\newblock An mcmc approach to classical estimation.
\newblock \textit{Journal of Econometrics} \textbf{115} 293--346.

\bibitem[{Fan et~al.(2014)Fan, Han and Liu}]{fan2014challenges}
\textsc{Fan, J.}, \textsc{Han, F.} and \textsc{Liu, H.} (2014).
\newblock Challenges of big data analysis.
\newblock \textit{National science review} \textbf{1} 293--314.

\bibitem[{Fan and Liao(2014)}]{fan2014endogeneity}
\textsc{Fan, J.} and \textsc{Liao, Y.} (2014).
\newblock Endogeneity in high dimensions.
\newblock \textit{Annals of statistics} \textbf{42} 872.

\bibitem[{Gautier and Tsybakov(2014)}]{gautier2014high}
\textsc{Gautier, E.} and \textsc{Tsybakov, A.} (2014).
\newblock High-dimensional instrumental variables regression and confidence
  sets.
\newblock Tech. rep., HAL.

\bibitem[{George and McCulloch(1997)}]{george1997approaches}
\textsc{George, E.~I.} and \textsc{McCulloch, R.~E.} (1997).
\newblock Approaches for bayesian variable selection.
\newblock \textit{Statistica sinica}  339--373.

\bibitem[{Hansen(1982)}]{hansen1982large}
\textsc{Hansen, L.~P.} (1982).
\newblock Large sample properties of generalized method of moments estimators.
\newblock \textit{Econometrica: Journal of the Econometric Society}
  1029--1054.

\bibitem[{Hastie et~al.(2015)Hastie, Tibshirani and
  Wainwright}]{hastie:etal:15}
\textsc{Hastie, T.}, \textsc{Tibshirani, R.} and \textsc{Wainwright, M.}
  (2015).
\newblock \textit{Statistical Learning with Sparsity: The Lasso and
  Generalizations}.
\newblock Chapman and Hall/CRC.

\bibitem[{Kato et~al.(2013)}]{kato2013quasi}
\textsc{Kato, K.} \textsc{et~al.} (2013).
\newblock Quasi-bayesian analysis of nonparametric instrumental variables
  models.
\newblock \textit{The Annals of Statistics} \textbf{41} 2359--2390.

\bibitem[{Liao et~al.(2011)Liao, Jiang et~al.}]{liao2011posterior}
\textsc{Liao, Y.}, \textsc{Jiang, W.} \textsc{et~al.} (2011).
\newblock Posterior consistency of nonparametric conditional moment restricted
  models.
\newblock \textit{The Annals of Statistics} \textbf{39} 3003--3031.

\bibitem[{Mitchell and Beauchamp(1988)}]{mitchell1988bayesian}
\textsc{Mitchell, T.~J.} and \textsc{Beauchamp, J.~J.} (1988).
\newblock Bayesian variable selection in linear regression.
\newblock \textit{Journal of the American Statistical Association} \textbf{83}
  1023--1032.

\bibitem[{Robert and Casella(2004)}]{robertetcasella04}
\textsc{Robert, C.~P.} and \textsc{Casella, G.} (2004).
\newblock \textit{Monte {C}arlo statistical methods}.
\newblock 2nd ed. Springer Texts in Statistics, Springer-Verlag, New York.

\bibitem[{Tierney(1994)}]{tierney94}
\textsc{Tierney, L.} (1994).
\newblock Markov chains for exploring posterior distributions.
\newblock \textit{Ann. Statist.} \textbf{22} 1701--1762.
\newblock With discussion and a rejoinder by the author.

\bibitem[{Yang et~al.(2016)Yang, Wainwright, Jordan
  et~al.}]{yang2016computational}
\textsc{Yang, Y.}, \textsc{Wainwright, M.~J.}, \textsc{Jordan, M.~I.}
  \textsc{et~al.} (2016).
\newblock On the computational complexity of high-dimensional bayesian variable
  selection.
\newblock \textit{The Annals of Statistics} \textbf{44} 2497--2532.

\end{thebibliography}
\end{document}